\documentclass[reqno]{amsart}
\usepackage{amsmath}
\usepackage{amsthm}
\usepackage{graphicx}
\usepackage{amsfonts}
\usepackage{amssymb}
\usepackage{enumerate}
\usepackage[top=1in, bottom=1in, left=1in, right=1in]{geometry}

\numberwithin{equation}{section}

\newtheorem{theorem}{Theorem}[section]
\newtheorem{lemma}[theorem]{Lemma}
\newtheorem{proposition}[theorem]{Proposition}
\newtheorem{corollary}[theorem]{Corollary}

\theoremstyle{definition}

\newtheorem{definition}[theorem]{Definition}

\newtheorem{remark}[theorem]{Remark}

\def\E{{\mathbb E}}
\def\R{{\mathbb R}}

\def\P{{\mathcal P}}

\def\H{{\mathcal H}}
\def\X{{\mathcal X}}

\def\G{{\mathcal G}}

\def\T{{\mathcal T}}

\def\F{{\mathcal F}}

\title{Liquidity, risk measures, and concentration of measure}
\author{Daniel Lacker}

\thanks{This material is based upon work supported by the National Science Foundation under Award No. DMS-1502980}

\address{\noindent Brown University, Division of Applied Mathematics, 182 George St, Providence, RI 02906}
\email{daniel\_lacker@brown.edu}

\begin{document}

\begin{abstract}
Expanding on techniques of concentration of measure, we develop a quantitative framework for modeling liquidity risk using convex risk measures. The fundamental objects of study are curves of the form $(\rho(\lambda X))_{\lambda \ge 0}$, where $\rho$ is a convex risk measure and $X$ a random variable, and we call such a curve a \emph{liquidity risk profile}. The shape of a liquidity risk profile is intimately linked with the tail behavior of the underlying $X$ for some notable classes of risk measures, namely shortfall risk measures. We exploit this link to systematically bound liquidity risk profiles from above by other real functions $\gamma$, deriving tractable necessary and sufficient conditions for \emph{concentration inequalities} of the form $\rho(\lambda X) \le \gamma(\lambda)$, for all $\lambda \ge 0$. These concentration inequalities admit useful dual representations related to transport inequalities, and this leads to efficient uniform bounds for liquidity risk profiles for large classes of $X$. On the other hand, some modest new mathematical results emerge from this analysis, including a new characterization of some classical transport-entropy inequalities. Lastly, the analysis is deepened by means of a surprising connection between time consistency properties of law invariant risk measures and the tensorization of concentration inequalities.
\end{abstract}

\maketitle

\section{Introduction}
The main goal of this paper is to develop a quantitative model of liquidity risk in terms of convex risk measures. As the title suggests, this is accomplished using ideas from the theory of concentration of measure. 
Conversely, we hope to illustrate that the well-developed convex duality for risk measures sheds new light on some problems of concentration, specifically pertaining to transport inequalities, although these modest results are not the main thrust of the paper.
Throughout the paper, a \emph{risk measure} is any convex functional $\rho : L^1 \rightarrow (-\infty,\infty]$ satisfying the following axioms:
\begin{enumerate}
\item Normalization: $\rho(0)=0$.
\item Cash additivity: $\rho(X + c) = \rho(X) + c$ for all $X \in L^1$, $c \in \R$.
\item Monotonicity: $\rho(X) \le \rho(Y)$ whenever $X,Y \in L^1$ with $X \le Y$ a.s.
\end{enumerate}
The space $L^1$ is defined relative to a fixed probability space $(\Omega,\F,P)$. Note that risk measures are more often assumed to be \emph{decreasing}, whereas we assume they are \emph{increasing}, which makes no difference beyond a sign change.

We think of a random variable $X$ as the value of a financial \emph{loss} (i.e., positive numbers are losses, negatives are gains), quoted in units of a safe or liquid asset, realized at the end of some fixed trading period. (Since our sign convention is opposite of the usual, we think of $X$ as a loss and not a gain.)
As usual, $\rho(X)$ quantifies the risk of the loss $X$, or more precisely the minimal capital (denominated in terms of a safe or risk-free reference instrument) that must be subtracted from the loss $X$ to make it acceptable. If $\rho(X) \le 0$, the position is \emph{acceptable}, while $\rho(X) > 0$ means the position is risky. Originally, risk measures were introduced by Artzner et al. \cite{artzner1999coherent} with an additional axiom of \emph{coherence}, meaning $\rho(\lambda X) = \lambda \rho(X)$ for all $\lambda \ge 0$ and $X \in L^1$. This ignores liquidity risk, in the sense that doubling a loss doubles its risk. For precisely this reason, convex risk measures were introduced in \cite{follmer-schied-convex,frittelli2002putting}, rendering the curve $(\rho(\lambda X))_{\lambda \ge 0}$ nonlinear. 
Nonetheless, explicit use of convex risk measures to quantify liquidity risk has been quite limited, and all such studies go well beyond the traditional setting by using more complex types of risk measures, such as liquidity-adjusted risk measures \cite{acerbi-scandolo-liquidity,weber-liquidityadjusted,jarrow-protter-liquidity} or set-valued risk measures \cite{jouini-meddeb-touzi,hamel-heyde-rudloff}.

In this paper we take the convexity axiom seriously by developing some quantitative tools with which classical convex risk measures can be used to model liquidity risk. We do so not by attaching to each position $X$ a \emph{single quantity} capturing the liquidity of the position, but rather by studying how to manipulate and bound curves of the form $(\rho(\lambda X))_{\lambda \ge 0}$. In this sense, liquidity to us is  infinite-dimensional, not one-dimensional.
Specifically, we undertake a thorough study of the functions $(\rho(\lambda X))_{\lambda \ge 0}$, where $X \in L^1$ is a given loss and $\rho$ a given risk measure. We call this function the \emph{liquidity risk profile} associated to $X$, in light of the fact that this curve describes exactly how the risk of a financial loss $X$ scales with its size. Of course, if the risk measure is coherent, then the linear liquidity risk profile $\rho(\lambda X) = \lambda \rho(X)$ is not very interesting. More generally, $\lambda \mapsto \rho(\lambda X)$ is always convex. Since $\rho$ is normalized, when $\lambda \in [0,1]$ we have $\rho(\lambda X) \le \lambda\rho(X)$, and when $\lambda > 1$ we have $\rho(\lambda X) \ge \lambda\rho(X)$. In particular, if a position $X$ is not acceptable (i.e., $\rho(X) > 0$), then the risk $\rho(\lambda X)$ scales super-linearly with the size of the position. This is quite natural, since in illiquid markets a constant price is typically unavailable for large volumes.

The proposed framework is intentionally ambiguous about the mechanisms behind liquidity, and we even remain agnostic about the very definition of the term. Rather, we find it more agreeable to model directly \emph{liquidity risk} as a concept distinct from any precise notion of market liquidity. Illiquidity could come from (or even be defined by) a wide variety of market frictions, such as a scarcity of counterparties or prohibitive search costs or transaction costs. No matter the source of illiquidity, we can agree that leverage is riskier in less liquid markets. Liquidity risk profiles, as we have defined them, indeed describe an interplay between leverage and risk, in a manner that adapts to any choice of risk measure. 
To distinguish between liquidity risk and market liquidity is no less reasonable than the widespread practice of abstracting the source of randomness in probabilistic models: The precise mechanisms behind the randomness are typically hidden in ``omega,'' that is in the dependence of a random variable $X$ on the underlying source of uncertainty, while only the distribution of $X$ is relevant for many modeling purposes. 

Using ideas from concentration of measure, we will derive tractable descriptions of and criteria for bounding liquidity risk profiles. Primarily we study \emph{concentration inequalities} of the form
\begin{align}
\rho(\lambda X) \le \gamma(\lambda), \text{ for all } \lambda \ge 0, \label{intro:main-inequality}
\end{align}
where $\gamma \ge 0$ is a given increasing convex function.
This inequality means that the capital required to cover the risk of a loss of $\lambda X$ is no more than $\gamma(\lambda)$, for each $\lambda \ge 0$. Values of $\gamma(\lambda)$ for \emph{large} $\lambda$ tell you how risky it is to leverage the position, whereas the behavior as $\lambda \downarrow 0$ tells you more about the marginal risk of $X$. 
Suppose for the moment that $\rho(Y) \ge \E Y$ for all $Y \in L^1$, which indeed is well known to hold if $\rho$ is law invariant (see Proposition \ref{pr:convexorder}). Then
\begin{align}
\rho(\lambda(X - \E X)) \ge 0, \text{ for all } \lambda \in \R. \label{intro:def:risklowerbound}
\end{align}
Thinking of the underlying $P$ as a pricing measure, $\E X$ is the (liquid) price of the position $X$, and $X - \E X$ is the payoff minus its price. The more relevant liquidity risk profile may then be that of $X - \E X$, which is the loss faced when selling the claim $X$ to an external party for the price of $\E X$;
unless equality holds in \eqref{intro:def:risklowerbound}, this incurs a nontrivial liquidity risk.  Hence, we focus on concentration inequalities of the form (rewritten using cash additivity)
\begin{align}
\rho(\lambda (X - \E X)) \le \gamma(\lambda) \quad \Leftrightarrow \quad \rho(\lambda X) \le \lambda \E X + \gamma(\lambda), \ \forall \lambda \ge 0. \label{intro:main-inequality2}
\end{align}
Of course, this is a special case of \eqref{intro:main-inequality}, simply with another choice of $\gamma$. It should be noted that functionals of the form $X \mapsto \rho(X - \E X)$ have themselves been studied extensively in recent years under the name \emph{deviation measures}, following Rockafellar et al. \cite{rockafellar2006generalized}, but we find the language of risk measures to be better suited to our purposes.

What are some reasonable and informative choices of $\gamma$? In light of \eqref{intro:def:risklowerbound}, we assume $\gamma$ is nonnegative, which will also simplify many calculations later.
Additionally, as $\lambda \downarrow 0$, we expect some form of continuity to enforce $\rho(\lambda X) \rightarrow \rho(0) = 0$. Thus, a sharp concentration inequality would have $\gamma(0)=0$, unless one is only concerned with controlling \emph{large} scalings of the loss $X$. The right-derivative
\begin{align}
\lim_{\lambda \downarrow 0}\lambda^{-1}\rho(\lambda(X - \E X)) \label{intro:rightderivative}
\end{align}
captures the marginal or transactional risk of selling $X$ (again at the price $\E X$). Indeed, if $\rho \ge \E$ on $L^1$ as in the previous paragraph, then $\rho(-X) \le \E X \le \rho(X)$ for all $X$, and we may think of the difference between the right-derivative in \eqref{intro:rightderivative} and the analogous left-derivative as a bid-ask spread.
These derivatives are zero precisely when liquidity risk vanishes with position size, in the sense that $\rho(\lambda X)$ behaves like $\lambda \E X$ at the first order. That is, when buying or selling a small amount of $X$, the risk should be roughly the price. With this in mind, we will pay special attention to the case $\gamma'(0)=0$.
If again $\rho \ge \E$ on $L^1$, and if $\gamma(0)=0$, then the bound \eqref{intro:main-inequality2} implies
\[
\E X \le \lim_{\lambda \downarrow 0}\lambda^{-1}\rho(\lambda X) \le \E X +  \gamma'(0).
\]
It was shown by Barrieu and El Karoui \cite{barrieu-elkaroui-infconvolution} that the decreasing limit $M\rho(X) := \lim_{\lambda \downarrow 0}\lambda^{-1}\rho(\lambda X)$ defines a coherent risk measure, which is naturally interpreted as a limit of vanishing liquidity risk. When $\gamma(0)=\gamma'(0)=0$, we see that if $X$ satisfies the concentration inequality \eqref{intro:main-inequality2} then $M\rho(X) = \E X$, or equivalently $M\rho(X - \E X) = 0$.

Before discussing the mathematical ideas any further, let us address some critiques leveled on convex risk measures by Acerbi and Scandolo in their influential recent paper \cite{acerbi-scandolo-liquidity}. First, they make a compelling case that convex risk measures are best interpreted as evaluating risk associated to \emph{marked-to-market portfolio values}, distinct from \emph{portfolio content}. Nonetheless, we hope to illustrate a meaningful theory of liquidity which is both simple and implementable, and which avoids recourse to the additional nonlinear \emph{value functions} of \cite{acerbi-scandolo-liquidity}. Moreover, we disagree somewhat with their primary critique:
\begin{quotation}
Violations [of positive homogeneity and subadditivity] are introduced at the level of the chosen risk measures and therefore they apply to all portfolios irrespective of their size or content. This in turn has the consequence that in the asymptotic limit of vanishing liquidity risk we will not be able to recover the fully coherent scheme, which we consider the appropriate one in this limit. \cite{acerbi-scandolo-liquidity}
\end{quotation}
This we contest on two grounds. First, as mentioned in the previous paragraph, it was shown in \cite{barrieu-elkaroui-infconvolution} that the (decreasing) limit $M\rho(X) = \lim_{\lambda \downarrow 0}\lambda^{-1}\rho(\lambda X)$ always defines a \emph{coherent} risk measure. This limit should indeed be interpreted as vanishing liquidity risk, as it draws out the marginal risk. Second, we point out that the liquidity risk profile $\lambda \mapsto \rho(\lambda X)$ depends quite sensitively on the loss $X$. Interpreting a random variable $X$ as a marked-to-market portfolio value, the precise structure of $X$ (namely its dependence on $\omega$) \emph{hides} but certainly does not \emph{remove} the dependence on the ``size or content'' of the underlying portfolio. This is even more clear when we introduce an \emph{initial position} $Y$, and study the liquidity risk profile $(\rho(Y + \lambda X) - \rho(Y))_{\lambda \ge 0}$ of $X$ relative to $Y$. We subtract $\rho(Y)$ for normalization purposes, noting that $X \mapsto \rho(X + Y) - \rho(Y)$ is again a risk measure by our definition. We will see in Corollary \ref{co:sensitiveinitialposition} that this liquidity risk profile depends heavily on the choice of $Y$ as well as $X$; in particular, if $\gamma \ge 0$ is nondecreasing and finite with $\gamma(0)=\gamma'(0)=0$, then the bound $\rho(Y + \lambda X) - \rho(Y) \le \gamma(\lambda)$ holds for all $\lambda \ge 0$ and all $Y \in L^1$ if and only if $X \le 0$ a.s.

Let us finally elucidate the first simple connection between risk measures and concentration of measure announced at the beginning of the paper. When $\rho$ is the \emph{entropic risk measure} $\rho(X) = \log\E[e^X]$, the concentration inequalities as we have defined them above are simply exponential bounds on the moment generating function or Laplace transform of the random variable $X$. For example, taking $\gamma(\lambda) = \sigma^2\lambda^2/2$ for some $\sigma > 0$, a random variable is called \emph{subgaussian} precisely when the it satisfies the inequality $\rho(\lambda X) \le \gamma(\lambda)$ for all $\lambda \in \R$. This can be characterized alternatively by a tail bound or an integral criterion:
\begin{enumerate}
\item There exist $c, \kappa > 0$ such that $P(|X| > t) \le c\exp(-\kappa t^2)$ for all $t > 0$.
\item There exists $c > 0$ such that $\E[\exp(c|X|^2)] < \infty$.
\end{enumerate}
A central example in this paper is the class of \emph{shortfall risk measures}, parameterized by a nondecreasing convex function $\ell \ge 0$ satisfying $\ell(0)=1$, defined by
\[
\rho(X) := \inf\left\{c \in \R : \E[\ell(X-c)] \le 1\right\}.
\]
When $\ell(x) = e^x$, this reduces to the entropic risk measure.
It is easy to check that $\rho(\lambda X) \le \gamma(\lambda)$ for all $\lambda \ge 0$ implies $P(X > t) \le 1/\ell(\gamma^*(t))$ for all $t > 0$, where $\gamma^*(t)=\sup_{\lambda \ge 0}(\lambda t - \gamma(\lambda))$. We will show in Theorem \ref{th:bigequivalence} that under certain additional assumptions on $\ell$ and $\gamma$, the following statements are equivalent in analogy with the subgaussian case:
\begin{enumerate}
\item There exists $c > 0$ such that $\E[\ell(\gamma^*(c |X|))] < \infty$.
\item There exist $c,\kappa > 0$ such that $P(|X| > t) \le c/\ell(\gamma^*(\kappa t))$ for all $t > 0$.
\item There exists $c > 0$ such that $\rho(\lambda X) \le \gamma(c|\lambda|)$ for all $\lambda \in \R$.
\end{enumerate}
The moment condition (1) is arguably the most tractable of these conditions, and we use it to derive a number of examples of random variables satisfying (1-3).
We then apply Theorem \ref{th:bigequivalence} to bound liquidity risk profiles of a large class of options in terms of just a few calls and puts. Namely, the liquidity risk profile of any Lipschitz function of $n$ underlying assets is bounded (uniformly in the choice of such function) by the maximum of $2n$ liquidity risk profiles: one call option and one put per underlying asset, each struck at the same price. See Section \ref{se:options} for details.
We also study briefly how another family of risk measures, namely optimized certainty equivalents, are connected to tail behavior in Section \ref{se:optcerttail}.

The connection between risk measures and concentration deepens with the study of \emph{tensorization}, which allows us to deduce concentration results for a product measure $\mu^n$ on a product space $E^n$ from concentration properties of $\mu$ on $E$. In other words, if we understand concentration properties of $X$ and $Y$, or more specifically various functions $f(X)$ and $g(Y)$ thereof, we can sometimes deduce concentration properties of combinations $h(X,Y)$ of the two. More generally, a financial position of the form $Y=f(X_1,\ldots,X_n)$ can be interpreted as a combination of $n$ underlying positions, and it is only natural to estimate the risk of $Y$ in terms of the risk of $X_1,\ldots,X_n$. Gaussian concentration is again a classical example: If $\rho(X) = \log\E[e^X]$ is the entropic risk measure, if $X_i$ are i.i.d. standard Gaussians, and if $f$ is $1$-Lipschitz on $\R^n$, then we know $\rho(\lambda(f(X_1,\ldots,X_n) - \E f(X_1,\ldots,X_n))) \le \lambda^2/2$ for all $\lambda \in \R$. This bound is efficient in that it is independent of the dimension $n$, and we can see this as an instance of diversification decreasing risk: For example, the average of $(x_1,\ldots,x_n)$ is a $n^{-1/2}$-Lipschitz function, and so
\[
\rho\left(\frac{\lambda}{n}\sum_{i=1}^nX_i\right) \le \lambda^2/2n, \text{ for all } \lambda \in \R.
\]
A well known concentration heuristic has financial meaning as well: If $X_i$ are independent, and if $Y=f(X_1,\ldots,X_n)$ does not depend \emph{too much} on any single one of the $X_i$, then we think of the risk $Y$ as diversified, and we expect $Y$ to concentration around its mean.

We study tensorization for \emph{law invariant} risk measures, that is satisfying $\rho(X) = \rho(Y)$ whenever $X$ and $Y$ have the same distribution. Classically, tensorization arguments rely on the chain rule for relative entropy, which reads
\[
H(\nu_1(dx)\nu_2(x,dy) | \mu_1(dx)\mu_2(x,dy)) = H(\nu_1 | \mu_1) + \int \nu_1(dx)H(\nu_2(x,\cdot) | \mu_2(x,\cdot)),
\]
for any (disintegrated) probability measures on any product of two measurable spaces. Here $H(\nu | \mu) = \int \log(d\nu/d\mu)d\nu$ if $\nu \ll \mu$, and it is infinite otherwise.
To extend these arguments we naturally seek a substitute for the chain rule for penalty functions (i.e., convex conjugates) of convex risk measures, and such a substitute is developed in \cite{lacker-lawinvariant}. It turns out, perhaps surprisingly, that a workable inequality form the chain rule is equivalent to a so-called \emph{time-consistency} property of a the risk measure, defined carefully in Section \ref{se:tensorization}. With this link in mind, we base our tensorization arguments on this time consistency property, instead of working with chain rules.
Our new tensorization results are modest, as it turns out that very few risk measures other than the entropic one have the right time consistency property. Nonetheless, this approach illustrates how the method must be altered in order to obtain sharper mathematical results, as will be explored in future work.

Lastly, we formulate risk measure concentration inequalities in terms of so-called \emph{transport inequalities}, which permits a connection with a more classical form of concentration in terms of enlargements of sets of metric measure spaces \cite{ledoux-concentration}. Financially, we will see that transport inequalities also provide bounds on liquidity risk profiles which are conveniently uniform across large families of losses $X$.
The success of this approach originated in the papers of Marton \cite{marton-bounding,marton-blowingup} and Talagrand \cite{talagrand-transportation}, and many more recent papers have ironed out the relations between transport inequalities, various other functional inequalities, concentration of measure, and even large deviations. We favor the perspective of Bobkov and Gotze \cite{bobkov-gotze}, who first noticed that following are equivalent for a probability measure $\mu$ on a complete separable metric space $(E,d)$ and a number $c > 0$:
\begin{enumerate}
\item For every $1$-Lipschitz function $f$ on $E$ and every $\lambda \ge 0$,  $\log\int e^{\lambda f}\,d\mu \le \lambda \int f\,d\mu + c\lambda^2/2$.
\item The transport inequality holds: $W_1(\mu,\nu) \le \sqrt{2cH(\nu | \mu)}$ for $\nu \ll \mu$, where $H$ is the relative entropy and $W_1$ is the Wasserstein distance (defined precisely in \eqref{def:wasserstein}).
\end{enumerate}
The easy proof of this fact is essentially an instance of the order-reversing property of convex conjugation, using only the duality between the entropic risk measure and the relative entropy (as well as the Kantorovich duality for the Wasserstein distance).
Borrowing these ideas, we derive dual forms of the concentration inequalities \eqref{intro:main-inequality}. Namely, \eqref{intro:main-inequality} is equivalent to
\[
\gamma^*(\E^Q[X]) \le \alpha(Q), \ \forall Q \ll P,
\]
where $\alpha$ is a so-called \emph{penalty function} for $\rho$, which will be defined precisely in Section \ref{se:riskmeasures}. A similar result holds for concentration inequalities with non-zero initial positions, that is, involving liquidity risk profiles of the form $\sup_{Y \in \Phi}(\rho(\lambda X + Y) - \rho(Y))_{\lambda \ge 0}$, and this provides some insights on the sensitivity of liquidity risk profiles to initial positions.

We extend some well known characterizations of transport inequalities to the risk measure setting, revisiting a generalization of the equivalence $(1) \Leftrightarrow (2)$ above which was essentially observed in \cite[Theorem 3.5]{gozlan-leonard-survey}, albeit without the language of risk measures. 
In fact, it is well known \cite{bobkov-gotze} that the conditions (1) and (2) above are also equivalent (up to a change in constant) to the existence of an exponential moment, or $\int\mu(dx)\exp(cd(x,x_0)^2) < \infty$ for some $x_0 \in E$, $c > 0$. We prove a similar integral criterion for our modified transportation inequalities involving shortfall risk measures, namely the finiteness of $\int \mu(dx)\ell(\gamma^*(cd(x,x_0)))$. 
An interesting mathematical point emerges here:
The integral criterion described above essentially depends only on the \emph{composition} $\ell \circ \gamma^*$, not on either function individually, whereas the transport inequalities seem to depend separately on $\ell$ and $\gamma^*$. Using this observation, we can show for example that the transportation inequality (2) above is equivalent (again up to a change in constant) to the seemingly weaker inequality
\begin{enumerate}
\item[(2')] For $\nu \ll \mu$, $W_1(\mu,\nu) \le c\sqrt{\log\left\|d\nu/d\mu\right\|_{L^\infty(\mu)}}$.
\end{enumerate}
A similar equivalence holds when $W_1$ is replaced by the quadratic  Wasserstein distance $W_2$; see Corollary \ref{co:tpequivalence}. Although we have no applications yet of this modest result (in particular we found no examples where it is easier to verify (2') than (2)), it seems interesting and nontrivial enough to warrant its brief discussion. Indeed, one reason transport inequalities are useful is because they can be verified directly in many important cases; see \cite{talagrand-transportation} for transport inequalities involving normal and exponential laws, \cite{cordero2002some} for log-concave densities, and \cite{feyel-ustunel} for infinite-dimensional Gaussian measures.

The paper is organized as follows. In Section \ref{se:riskmeasures}, we discuss the relevant background material on convex risk measures. Section \ref{se:integralcriteria} details the connection between concentration inequalities of the form \eqref{intro:main-inequality}, tail bounds, and moment estimates, mostly focusing on the class of shortfall risk measures. An application to options is given in Section \ref{se:options}, and several examples illustrating the theory are provided in Section \ref{se:examples1}. Section \ref{se:tensorization} discusses some results on tensorization. Finally, Section \ref{se:dualinequalities} turns to more abstract properties of concentration inequalities and their duals as well as connections with a more traditional form of concentration of measure on a metric space. Notable here is a new characterization of transport inequalities (Corollary \ref{co:tpequivalence}). Some additional examples of dual inequalities and possibilities for future research are discussed briefly in Section \ref{se:examples}, with a curious example coming from \emph{martingale optimal transport}, initiated recently in \cite{beiglbock-henrylabordere-penkner}. The appendix \ref{ap:proof} is devoted to the proof of Theorem \ref{th:integralcriterion}.

\section{Risk measure preliminaries} \label{se:riskmeasures}
First, let us fix some notation. Throughout the paper, $(\Omega,\F,P)$ is a fixed probability space. Abbreviate $L^p=L^p(\Omega,\F,P)$ as usual for the set of (equivalence classes of) $p$-integrable (or essentially bounded if $p=\infty$) real-valued measurable functions on $\Omega$. Let $\P(\Omega)$ denote the set of probability measures on $(\Omega,\F)$, and let $\P_P(\Omega)$ denote the subset consisting of those measures which are absolutely continuous with respect to $P$. For $q \in [1,\infty]$, let $\P_P^q(\Omega)$ denote the set of $Q \in \P_P(\Omega)$ with $dQ/dP \in L^q$. Given $Q \in \P(\Omega)$, we write $\E^Q$ for expectation under $Q$; the notation $\E$ is reserved for expectation under the reference measure $P$, although we will occasionally write $\E^P$ for emphasis.

Let $\X$ be a (linear) subspace of $L^1$ with $L^\infty \subset \X$. We will work mostly with $\X = L^1$ and occasionally with $\X = L^\infty$.
To us, a \emph{risk measure on $\X$} is a functional $\rho : \X \rightarrow (-\infty,\infty]$ satisfying
\begin{enumerate}
\item[(R1)] Monotonicity: If $X,Y \in \X$ and $X \le Y$ a.s. then $\rho(X) \le \rho(Y)$.
\item[(R2)] Cash additivity: If $X \in \X$ and $c \in \R$ then $\rho(X + c) = \rho(X) + c$.
\item[(R3)] Normalization: $\rho(0)=0$.
\item[(R4)] Convexity: If $X,Y \in \X$ and $t \in (0,1)$ then $\rho(tX + (1-t)Y) \le t\rho(X) + (1-t)\rho(Y)$.
\end{enumerate}
We say $X \in L^1$ is \emph{acceptable} if $\rho(X) \le 0$. 
An important additional property is satisfied by many but not all risk measures, and this is related to the dual representations of risk measures:
\begin{enumerate}
\item[(R5)] Fatou property: If $X,X_n,Y \in \X$ satisfy $|X_n| \le Y$ and $X_n \rightarrow X$ a.s., then $\rho(X) \le \liminf_{n\rightarrow\infty}\rho(X_n)$.
\end{enumerate}
We will work mostly with risk measures on $L^1$. Naturally, if $(\Omega',\F',P')$ is another probability space, a \emph{risk measure on $L^1(\Omega',F',P')$} is defined in the obvious way. Note that as a consequence of convexity and normalization, for any $X$ we have
\begin{align}
\rho(\lambda X) \le \lambda X \text{ for } \lambda \le 1, \quad \rho(\lambda X) \ge \lambda X \text{ for } \lambda \ge 1. \label{def:convexityinequality}
\end{align}

The convex conjugates of risk measures are central to the dual inequalities of Section \ref{se:dualinequalities}. Given a risk measure $\rho$ on $L^p$ for some $p \in [1,\infty]$ with conjugate exponent $q = p/(p-1)$, a function $\alpha : \P_P^q(\Omega) \rightarrow [0,\infty]$ is called a \emph{penalty function for $\rho$} if
\begin{align}
\rho(X) = \sup_{Q \in \P_P^q(\Omega)}\left(\E^Q[X] - \alpha(Q)\right), \ \text{ for all } X \in L^p. \label{def:penalty}
\end{align}
In other words, $\rho$ is the convex conjugate of the function defined on $L^q$ to equal $\alpha$ on $\{Z \in L^q : Z \ge 0, \ \E Z = 1\}$ (identied with the set $\P^q_P(\Omega)$) and to equal $\infty$ elsewhere. Note that $\P_P^1(\Omega) = \P_P(\Omega)$. We say a penalty function $\alpha$ is the \emph{minimal penalty function} for $\rho$ if every other penalty function $\alpha'$ satisfies $\alpha' \ge \alpha$. The existence of penalty functions is by now well understood:\footnote{If we permit finitely additive measures in the dual formula \eqref{def:penalty}, then every risk measure has a penalty function. While some of the arguments of this paper may work in this context, none of our examples require this level of generality.}

\begin{theorem}[Theorem 4.31 of \cite{follmer-schied-book}, Theorem 3.4 of \cite{kaina-ruschendorf}] \label{th:follmerschied}
Let $p \in [1,\infty]$, and let $q = p/(p-1)$ denote the conjugate exponent.
The following are equivalent for a risk measure $\rho$ on $L^p$:
\begin{enumerate}
\item $\rho$ has the Fatou property.
\item $\rho$ is lower semicontinuous with respect to $\sigma(L^p,L^q)$.
\item There exists a penalty function for $\rho$.
\item The convex function $\alpha : \P_P^q(\Omega) \rightarrow [0,\infty]$ defined by
\begin{align}
\alpha(Q) : = \sup_{X \in L^p}\left\{\E^Q[X] - \rho(X)\right\} = \sup\left\{\E^Q[X] : X \in L^p, \ \rho(X) \le 0\right\} \label{def:minimalpenalty}
\end{align}
is a penalty function for $\rho$.
\end{enumerate}
Moreover, the penalty function given in \eqref{def:minimalpenalty} is minimal, in the sense that every other penalty function for $\rho$ dominates it. 
\end{theorem}

In this paper we work almost exclusively with $\X = L^1$, but we briefly discuss the $\X = L^\infty$ below in the context of extensions of law invariant risk measures.
Risk measures on other spaces $\X$ do indeed appear in applications, particularly when $\X$ is an Orlicz space \cite{cheridito-li-orlicz,biagini-frittelli}. At such a level of generality, convex duality and the connections between lower semicontinuity and the Fatou property are both more delicate but are developed carefully in \cite{biagini-frittelli}. We will not worry about these generalities here, since all of our examples are (or extend to) risk measures on $L^1$.
The most important examples of risk measures in this paper happen to be \emph{law invariant}, in the sense that $\rho(X) = \rho(Y)$ whenever $X$ and $Y$ have the same law.
Law invariant risk measure possess some nice additional structure, highlighted by the results of Jouini, Touzi, and Schachermayer \cite{jouini-touzi-schachermayer} as well as Filipovi\'c and Svindland \cite{filipovic-svindland}:

\begin{theorem}[Theorem 2.1 of \cite{jouini-touzi-schachermayer}, Proposition 1.1 of \cite{svindland-continuity}] \label{th:JTS}
Every law invariant risk measure on $L^\infty$ has the Fatou property.
\end{theorem}

\begin{theorem}[Theorem 2.2 of \cite{filipovic-svindland}] \label{th:filipovic-svindland}
A law invariant risk measure $\rho$ on $L^\infty$ admits a unique extension to $L^1$. Precisely, there exists a risk measure $\bar{\rho}$ on $L^1$ such that $\bar{\rho} = \rho$ on $L^\infty$, and it satisfies the duality relation
\[
\bar{\rho}(X) = \sup\left\{\E^Q[X] - \alpha(Q) : Q \in \P_P(\Omega), \ \frac{dQ}{dP} \in L^\infty\right\},
\]
where $\alpha$ is the minimal penalty function of $\rho$, defined on $\P_P(\Omega)$ by
\[
\alpha(Q) := \sup_{X \in L^\infty}\left(\E^Q[X] - \rho(X)\right).
\]
\end{theorem}

This tells us that when dealing with law-invariant risk measures, we may without loss of generality assume they are defined on all of $L^1$. One final notably result is that a law invariant risk measure is automatically increasing with respect to convex order:

\begin{proposition}[Corollary 4.65 of \cite{follmer-schied-book}, Theorem 2.1 of \cite{svindland-convexorder}] \label{pr:convexorder}
Suppose $\rho$ is a law invariant risk measure on $L^\infty$, extended canonically to $L^1$. If $X,Y \in L^1$ satisfy $\E[\phi(X)] \le \E[\phi(Y)]$ for every increasing convex function $\phi$, then $\rho(X) \le \rho(Y)$. In particular, $\rho(\E[X | \G]) \le \rho(X)$ for every $X \in L^1$ and every $\sigma$-field $\G \subset \F$.
\end{proposition}

\section{Concentration inequalities and an integral criterion} \label{se:integralcriteria}

We now develop some necessary and sufficient conditions for \emph{concentration inequalities} of the form
\begin{align}
\rho(\lambda X) \le \gamma(\lambda), \text{ for all } \lambda \ge 0, \label{def:concentrationprototype}
\end{align}
where $\gamma : [0,\infty) \rightarrow [0,\infty]$ is a given nondecreasing function, $\rho$ a given risk measure on $L^1$, and $X \in L^1$.
We will extend the domain of $\gamma$ to $\R$ by setting $\gamma(\lambda) = \infty$ for all $\lambda < 0$. Then its convex conjugate is $\gamma^*(t) = \sup_{\lambda \ge 0}(t\lambda - \gamma(\lambda))$ and satisfies $\gamma^*(t) = -\gamma(0)$ for $t \le 0$. In a sense, it is without loss of generality that we assume $\gamma$ is convex and lower semicontinuous: Suppose that $\rho$ has the Fatou property (i.e., is lower semicontinuous by Theorem \ref{th:follmerschied}), and suppose that $\rho(\lambda X) \le \gamma(\lambda)$ for all $\lambda \ge 0$, where $\gamma$ is \emph{not necessarily convex}. Then, since the biconjugate $\gamma^{**}$ is the largest convex lower semicontinuous minorant of $\gamma$, we conclude that $\rho(\lambda X) \le \gamma^{**}(\lambda)$ for all $\lambda \ge 0$. It will soon become clear why we assume $\gamma \ge 0$, although this accepts some loss of generality.

\begin{definition}
A \emph{shape function} is any nondecreasing, convex, and lower semicontinuous function $\gamma : [0,\infty) \rightarrow [0,\infty]$ satisfying $\gamma(0) < \infty$.
\end{definition}

In the rest of this section, we specialize to shortfall risk measures and optimized certainty equivalents. We find that liquidity risk profiles for these types of risk measures encode useful information about the tail behavior of random variables, much like the moment generating function.

\subsection{Tails and shortfall risk measures} \label{se:tails-shortfall}

The main risk measures of interest are the \emph{entropic} risk measure $\rho(X) = \log\E[e^{X}]$ and its generalization to \emph{shortfall} risk measures:

\begin{definition} \label{def:lossfunction}
A \emph{loss function} is a convex and nondecreasing function $\ell : \R \rightarrow [0,\infty)$ satisfying $\ell(0) = 1 < \ell(x)$ for all $x > 0$. The \emph{shortfall risk measure} corresponding to $\ell$ is given by
\begin{align}
\rho(X) = \inf\{c \in \R : \E[\ell(X-c)] \le 1\}, \ X \in L^1. \label{def:shortfall}
\end{align}
\end{definition}

According to \cite[Theorem 4.106]{follmer-schied-book} the minimal penalty function of a shortfall risk measure $\rho$ is
\begin{align}
\alpha(Q) = \inf_{t > 0}\frac{1}{t}\left(1 + \E^P[\ell^*\left(tdQ/dP\right)]\right), \label{def:shortfallentropy}
\end{align}
where $\ell^*(x) = \sup_{y \in \R}(xy - \ell(y))$ is the convex conjugate. For example, taking $\ell(t) = (1+t)^+$, the conjugate is
\[
\ell^*(t) = \begin{cases}
-t &\text{if } t \in [0,1] \\
\infty &\text{otherwise}
\end{cases},
\]
and the corresponding minimal penalty function is $\alpha(Q) = \|dQ/dP\|_{L^\infty(P)} - 1$. On the other hand, if $p \in (1,\infty)$ and $q = p/(p-1)$ is the conjugate exponent, the conjugate of
\[
\ell(t) = ((1+t)^+)^p \quad \text{ is } \quad \ell^*(t) = \begin{cases}
\frac{p}{q}\left(\frac{t}{p}\right)^q - t &\text{if } t \ge 0 \\
\infty &\text{otherwise}
\end{cases},
\]
and the corresponding penalty function is exactly 
\[
\alpha(Q) = \|dQ/dP\|_{L^q} - 1 = \left[
\E^P[(dQ/dP)^q\right]^{1/q} - 1.
\]
Indeed, setting $c = (p/qp^q)\E[(dQ/dP)^q]$, we have from \eqref{def:shortfallentropy}
\[
\alpha(Q) = \inf_{t > 0}\left(\frac{1}{t} - 1 + ct^{q-1} \right).
\]
The infimum is attained by $t = [c(q-1)]^{1/q}$, and the infimum equals
\[
[c(q-1)]^{1/q} - c[c(q-1)]^{-(q-1)/q} - 1 = \kappa_pc^{1/q}-1,
\]
where $\kappa_p = (q-1)^{1/q} + (p-1)^{1/p}$. But
\[
\kappa_pc^{1/q} = \frac{(p-1)^{1/p} + (q-1)^{1/q}}{p^{1/p}q^{1/q}}\E^P[(dQ/dP)^q]^{1/q} = \E^P[(dQ/dP)^q]^{1/q}.
\]

\begin{proposition} \label{pr:shortfall-tail}
Let $\ell$ be a loss function with corresponding shortfall risk measure $\rho$ defined as in \eqref{def:shortfall}, and let $\gamma$ be a shape function. If $\rho(\lambda X) \le \gamma(\lambda)$ for all $\lambda > 0$, then
\[
P(X > t) \le \frac{1}{\ell(\gamma^*(t))}, \quad \forall t > 0.
\]
\end{proposition}
\begin{proof}
Note that by continuity of $\ell$ and monotone convergence, the infimum is always attained in the definition of $\rho$ in \eqref{def:shortfall}. In particular,
\[
\E[\ell(Y-\rho(Y))] \le 1, \ \forall Y \in L^1.
\]
Cash additivity implies $\rho(\lambda X - \gamma(\lambda)) \le 0$, which in turn implies
\[
\E\left[\ell(\lambda X - \gamma(\lambda))\right] \le 1.
\]
Since $\ell$ is nondecreasing, so is $x \mapsto \ell(\lambda x - \gamma(\lambda))$ for each $\lambda > 0$. By Markov's inequality,
\[
P(X > t) \le P\left[\ell(\lambda X - \gamma(\lambda)) \ge \ell(\lambda t - \gamma(\lambda))\right] \le \frac{1}{\ell(\lambda t - \gamma(\lambda))}.
\]
Optimizing over $\lambda > 0$ completes the proof.
\end{proof}

\subsection{Tails and optimized certainty equivalents} \label{se:optcerttail}

Shortfall risk measures are not the only ones which yield some control over tail behavior. Let $\phi : \R \rightarrow \R$ be convex and nondecreasing, satisfying $\phi^*(1) = \sup_{x \in \R}(x - \phi(x)) = 0$. The associated \emph{optimized certainty equivalent} is defined as in \cite{bental-teboulle-1986,bental-teboulle-2007} by 
\begin{align}
\rho(X) := \inf_{m \in \R}\left(\E[\phi(m+X)] - m\right), \ X \in L^1. \label{def:optimizedcertainty}
\end{align}
The minimal penalty function is the $\phi^*$-divergence,
\[
\alpha(Q) = \E^P[\phi^*(dQ/dP)].
\]
Some notable examples are as follows:
\begin{enumerate}
\item If $\phi(t) = e^{t-1}$, the conjugate is $\phi^*(t) = t\log t$ on $t \ge 0$. In this case, we have $\rho(X) = \log\E[e^X]$, and $\alpha$ is the usual relative entropy.
\item If $p \in (1,\infty)$ has conjugate exponent $q=p/(p-1)$, then $\phi(t) = 1 + (p/q)(t/p)^q1_{t \ge 0}$ has conjugate $\phi^*(t) = t^p -1$ on $t \ge 0$, and $\alpha(Q) = \|dQ/dP\|_{L^p}^p - 1$.
\item If $p \in (1,\infty)$ has conjugate exponent $q=p/(p-1)$, then $\phi(t) = (q-1) + (t/q)^q1_{t \ge 0}$ has conjugate $\phi^*(t) = (t^p -1)/(p-1)$ on $t \ge 0$, and $\alpha(Q) = (\|dQ/dP\|_{L^p}^p - 1)/(p-1)$ is known as R\'enyi's divergence of order $p$.
\end{enumerate}
The final example of R\'enyi's divergence was studied in connection with transport inequalities and concentration inequalities in recent papers of Bobkov and Ding \cite{bobkov-ding,ding2014wasserstein}.

\begin{proposition} \label{pr:optcerttail}
Suppose $\gamma$ is a nonnegative and strictly increasing shape function. Let $\phi$ be a strictly positive and nondecreasing convex function satisfying $\phi^*(1)=0$, and define the corresponding optimized certainty equivalent as in \eqref{def:optimizedcertainty}. If $X \in L^1$ satisfies $\rho(\lambda X) \le \gamma(\lambda)$ for all $\lambda > 0$, then 
\[
P(X > s) \le \sup_{m > 0} \frac{m}{\phi(m + \gamma^*(s))}, \text{ for } s > 0.
\]
\end{proposition}
\begin{proof}
Fix $\epsilon > 0$.
Since $\rho(\lambda X - \gamma(\lambda)) \le 0$, it holds by definition of $\rho$ that for all $\lambda \ge 0$ there exists $m_\lambda \in \R$ such that
\[
m_\lambda \ge \E\left[\phi\left(m_\lambda + \lambda X - \gamma(\lambda)\right)\right] - \epsilon.
\]
Since $\phi > 0$, note that $m_\lambda > -\epsilon$.
Since $\phi$ is nondeincreasing, using Markov's inequality we get for any $s,\lambda > 0$
\begin{align*}
P(X > s) &\le P\left\{\phi\left(m_\lambda + \lambda X - \gamma(\lambda)\right) \ge \phi\left(m_\lambda + \lambda s - \gamma(\lambda)\right)\right\} \\
	&\le \frac{\E\left[\phi\left(m_\lambda + \lambda X - \gamma(\lambda)\right)\right]}{\phi\left(m_\lambda + \lambda s - \gamma(\lambda)\right)} \\
	&\le \frac{m_\lambda + \epsilon}{\phi\left(m_\lambda + \lambda s - \gamma(\lambda)\right)} \\
	&\le \sup_{m > -\epsilon}\frac{m + \epsilon}{\phi\left(m + \lambda s - \gamma(\lambda)\right)}.
\end{align*}
By definition of $\gamma^*$, we can find $\lambda$ so that $\lambda s - \gamma(\lambda) > \gamma^*(s) - \epsilon$. Then, again since $\phi$ is nondecreasing, 
\begin{align*}
P(X > s) &\le \sup_{m > -\epsilon}\frac{m + \epsilon}{\phi\left(m + \gamma^*(s) - \epsilon\right)} = \sup_{m > -2\epsilon}\frac{m + 2\epsilon}{\phi\left(m + \gamma^*(s)\right)} \\
	&\le \sup_{m > -2\epsilon}\frac{m}{\phi\left(m + \gamma^*(s)\right)} + 2\epsilon\sup_{m > -2\epsilon}\frac{1}{\phi\left(m + \gamma^*(s)\right)} \\
	&\le \sup_{m > -2\epsilon}\frac{m}{\phi\left(m + \gamma^*(s)\right)} + \frac{2\epsilon}{\phi\left(-2\epsilon + \gamma^*(s)\right)}
\end{align*}
Sending $\epsilon \downarrow 0$ completes the proof, since $\phi > 0$ implies
\[
\sup_{m > -2\epsilon}\frac{m}{\phi\left(m + \gamma^*(s)\right)} = \sup_{m > 0}\frac{m}{\phi\left(m + \gamma^*(s)\right)}.
\]
\end{proof}

Let us return briefly to some of the examples above. If $\phi(t) = e^{t-1}$, then one easily checks that $\sup_{m > 0} m/\phi(m + a) = \exp(-a)$, and Proposition \ref{pr:optcerttail} recovers the usual result on subgaussian concentration. In the second example (with $p=2$), that is with $\phi(t) = t^2/4 + 1$ on $t \ge 0$ and $\phi(t) =1$ on $t < 0$, we have
\[
\sup_{m > 0} \frac{m}{\phi(m + a)} = \frac{2\sqrt{a^2 + 4}}{4 + a^2 + a\sqrt{a^2 + 4}}, \text{ for } a > 0.
\]

\subsection{An integral criterion for shortfall risk measures}

We have just seen that certain bounds on the curve $t \mapsto P(X > t)$ are necessary for certain concentration inequalities, for certain risk measures. This section explores some \emph{sufficient} conditions, in the form of what are traditionally called \emph{integral criteria}, for the case of shortfall risk measures.
A classical example is the integral criterion for subgaussianity: a random variable is subgaussian (meaning $\log\E[\exp(\lambda X)] \le c\lambda^2/2$ for all $\lambda \in \R$ and some $c > 0$) if and only if there exists $a > 0$ such that $\E[\exp(a|X|^2)] < \infty$.
Let us now identify a class of $\ell$ and $\gamma$ for which an analogous sort of converse to Proposition \ref{pr:shortfall-tail} holds.

\begin{definition} \label{def:LGamma}
Let $\mathrm{L}\Gamma$ denote the class of functions $(\ell,\gamma)$, where $\ell$ is a loss function and $\gamma$ is a shape function, such that the following hold:
\begin{enumerate}
\item $\gamma(x) < \infty$ for some $x > 0$.
\item $\lim_{x \rightarrow \infty}\ell(\gamma^*(x))/x^2 = \infty$.
\item One of the following conditions holds:
\begin{enumerate}
\item $\gamma(0) > 0$, and $\gamma$ is nonconstant.
\item $\gamma(0)=0$, $\ell$ is continuously differentiable with $\ell'(0) > 0$, $\gamma$ is continuously differentiable in a neighborhood of the origin with $\gamma'(0)>0$, and
\[
\lim_{x \rightarrow \infty}x\ell'(x)/\ell(\gamma^*(x)) = 0.
\]
\item $\gamma(0)=\gamma'(0)=0$, $\ell$ is twice continuously differentiable with $\ell'(0) > 0$,  $\ell''$ is nondecreasing, $\gamma$ is twice continuously differentiable in a neighborhood of the origin with $\gamma''(0)>0$, and 
\[
\lim_{x \rightarrow \infty}x^2\ell''(x)/\ell(\gamma^*(x)) = 0.
\]
\end{enumerate}
\end{enumerate}
In conditions (3b) and (3c), the \emph{neighborhood} is relative to $[0,\infty)$ and the derivatives at $0$ exist as right-limits; that is, (3b) requires $\gamma$ to be continuously differentiable on $(0,a)$ for some $a > 0$, and $\gamma'(0) = \lim_{\delta\downarrow 0}\gamma'(\delta) > 0$.
\end{definition}

\begin{theorem} \label{th:integralcriterion}
Suppose that $(\ell,\gamma) \in \mathrm{L}\Gamma$. Let $\kappa,M > 0$. There exists $n > 0$ such that, if $\Phi \subset L^1(P)$ satisfies 
\begin{align}
\sup_{X \in \Phi}\E[\ell(\gamma^*(\kappa X^+)] \le M, \label{th:integralcriterion-hypothesis}
\end{align}
then we have $\rho(\lambda(X - \E X)) \le \gamma(n \lambda)$ for all $\lambda \ge 0$ and $X \in \Phi$.
\end{theorem}

The proof of Theorem \ref{th:integralcriterion} is in the appendix. To complete the picture, we give a name to one more condition that permits an integral criterion to be deduced from a tail bound:

\begin{definition} \label{def:classH}
Let $\H$ denote the class of a functions $h : [0,\infty) \rightarrow [1,\infty)$ which are nondecreasing and absolutely continuous, and for which there exists $c > 0$ such that
\begin{align}
\int_0^\infty\frac{h'(t)}{h(t/c)}dt < \infty. \label{eq:classH}
\end{align}
\end{definition}

\begin{lemma} \label{le:tailtointegral}
Let $h \in \H$, and let $c > 0$ be such that \eqref{eq:classH} holds. If $\Phi \subset L^1$ is such that
\[
\sup_{X \in \Phi}P\left(X > t\right) \le \frac{C}{h(t)}, \quad \forall t > 0, \text{ for some } C > 0,
\]
then $\sup_{X \in \Phi}\E[h(cX^+)] < \infty$.
\end{lemma}
\begin{proof}
Since $h$ is nondecreasing and $h \ge 1$, for each $X \in \Phi$ we have
\begin{align*}
\E[h(cX^+)] &= 1 + \int_1^\infty P(h(cX^+) > t) dt \le 1 + \int_0^\infty P(cX^+ > t)h'(t) dt \\
	&\le 1 + C\int_0^\infty \frac{h(t)}{h(t/c)} dt.
\end{align*}
\end{proof}

Combining the preceding results, along with a result (Proposition \ref{pr:dualinequality}) to be proven in the next section, we get the following result.

\begin{theorem} \label{th:bigequivalence}
Let $\ell$ be a loss function and $\gamma$ a shape function, and let $\rho$ be the shortfall risk measure corresponding to $\ell$. Suppose $\Phi \subset L^1$ satisfies $\E X = 0$ for all $X \in \Phi$. Consider the following statements:
\begin{enumerate}
\item There exists $c > 0$ such that, for all $\lambda \ge 0$,
\[
\sup_{X \in \Phi}\rho(\lambda X) \le \gamma(c\lambda).
\]
\item There exists $c > 0$ such that, for all $Q \in \P_P^\infty(\Omega)$,
\[
\gamma^*\left(c\sup_{X \in \Phi}\E^Q[X]\right) \le \inf_{t > 0}\frac{1}{t}\left\{1 + \E^P\left[\ell^*\left(t\frac{dQ}{dP}\right)\right]\right\}.
\]
\item There exists $c > 0$ such that, for all $t > 0$,
\[
\sup_{X \in \Phi}P(X > t) \le \frac{1}{\ell(\gamma^*(ct))}.
\]
\item There exists $c > 0$ such that
\[
\sup_{X \in \Phi}\E[\ell(\gamma^*(cX^+))] < \infty.
\]
\end{enumerate}
We have the implications $(1) \Leftrightarrow (2) \Rightarrow (3) \Leftarrow (4)$. If $\ell \circ \gamma^* \in \H$, then $(3) \Rightarrow (4)$. Lastly, if $(\ell,\gamma) \in \mathrm{L}\Gamma$, then $(4) \Rightarrow (1)$.
\end{theorem}
\begin{proof}
The equivalence of (1) and (2), with the same constant $c$, follows quickly from Proposition \ref{pr:dualinequality} of the next section. We saw in Proposition \ref{pr:shortfall-tail} that (1) implies (3). Since $\ell \circ \gamma^*$ is nondecreasing, it follows easily from Markov's inequality that (4) implies (3). If $(\ell,\gamma) \in \mathrm{L}\Gamma$, Theorem \ref{th:integralcriterion} says that (4) implies (1). If $\ell \circ \gamma^* \in \H$, it follows from Lemma \ref{le:tailtointegral} that (3) implies (4).
\end{proof}

\begin{remark} \label{re:meanzero}
If $\Phi \subset L^1$ in Theorem \ref{th:bigequivalence} does not consist solely of mean-zero random variables, then we may apply Theorem \ref{th:bigequivalence} with $\Phi$ replaced by $\Phi'= \{X - \E X : X \in \Phi\}$. In this sense, Theorem \ref{th:bigequivalence} speaks to the concentration of random variables \emph{around their means}. If $\Phi$ is symmetric in the sense that $\Phi = - \Phi$, then some changes can be made: Condition (1) can be replaced with $\rho(\lambda X) \le \gamma(c|\lambda|)$ for all $\lambda \in \R$, condition (3) with $P(|X| > t) \le 2/\ell(\gamma^*(ct))$ for $t > 0$, and condition (4) with $\sup_{X \in \Phi}\E[\ell(\gamma^*(c|X|))] < \infty$.
\end{remark}

\begin{remark}
There seems to be plenty of room for improvement in Theorem \ref{th:integralcriterion}, namely expanding the class $\mathrm{L}\Gamma$ of permissible $(\ell,\gamma)$. We were able to prove the following more abstract integral criterion, but only under the restrictive additional assumption that either $\gamma(0) >0$ or $\gamma'(0) > 0$:
Let $\rho$ be a risk measure on $L^1$ with the Fatou property, and let $\gamma$ be a shape function satisfying $\gamma(x) < \infty$ for some $x > 0$. Let $\Phi \subset L^1$ satisfy the following:
\begin{enumerate}[\quad\quad (i)]
\item For all $X \in \Phi$, $\lim_{\lambda \downarrow 0}\lambda^{-1}\rho(\lambda X) = \E X$.
\item $\Phi$ is uniformly integrable.
\item There exists $\kappa > 0$ such that $\sup_{X \in \Phi}\rho(\gamma^*(\kappa X^+)) < \infty$.
\end{enumerate}
Then there exists $c > 0$ such that $\rho(\lambda (X - \E X)) \le \gamma(c\lambda)$ for all $\lambda \ge 0$ and $X \in \Phi$. In certain cases it can be shown, for example, that for a shortfall risk measure, the moment bound of Theorem \ref{th:bigequivalence}(4) is enough to ensure that these conditions (i-iii) all hold, and a similar result is available for optimized certainty equivalents. The proof of this abstract integral criterion, like that of Theorem \ref{th:integralcriterion}, depends on a bound which essentially comes from Taylor's theorem; condition (i) is equivalent to
\[
\frac{d}{d\lambda}\rho(\lambda X)|_{\lambda = 0} = \E X, \text{ for } X \in \Phi,
\]
and this crucially allows us to bound the first order term in the expansion of $\lambda \mapsto \rho(\lambda X)$ around $\lambda = 0$. If $\gamma(0)=\gamma'(0)=0$, then this proof technique requires some information on the \emph{second} derivative of $\rho(\lambda X)$ near $\lambda = 0$, which is generally quite hard to obtain.
\end{remark}

As a consequence of Theorem \ref{th:bigequivalence}, we obtain a simplified integral criterion for a uniform concentration result, foreshadowing somewhat the discussion of transport inequalities in Section \ref{se:transportationinequalities}. When $\Omega$ is a metric space, we write $\mathrm{Lip}_1(\Omega)$ for the set of $1$-Lipschitz functions on $\Omega$, i.e., the set of $f : \Omega \rightarrow \R$ satisfying $|f(x)-f(y)| \le d(x,y)$ for all $x,y\in \Omega$. We write also $\P^1(\Omega)$ for the set of $P \in \P(\Omega)$ with $\int_\Omega d(x,x_0)P(dx) < \infty$, for some (equivalently, for every) $x_0 \in \Omega$.

\begin{corollary} \label{co:t1-shortfall}
Suppose $(\Omega,d)$ is a complete, separable metric space. Let $\ell$ be a loss function and $\gamma$ a shape function, and let $\rho$ denote the shortfall risk measure corresponding to $\ell$. Assume $P \in \P^1(\Omega)$. Consider the following statements:
\begin{enumerate}
\item There exists $c > 0$ such that for each $f \in \mathrm{Lip}_1(\Omega)$ and each $\lambda \in \R$ we have
\[
\rho\left(\lambda f\right) \le \gamma(|\lambda|) + \lambda\E f.
\]
\item There exists $c > 0$ such that for some (equivalently, for every) $x_0 \in \Omega$,
\[
\int_\Omega\ell(\gamma^*(cd(x,x_0)))P(dx) < \infty.
\]
\end{enumerate}
If $(\ell,\gamma) \in \mathrm{L}\Gamma$ (see Definition \ref{def:LGamma}), then $(2) \Rightarrow (1)$. If $\ell \circ \gamma^* \in \H$ (see Definition \ref{def:classH}), then $(1) \Rightarrow (2)$.
\end{corollary}
\begin{proof}
Since $x \mapsto d(x,x_0)$ is $1$-Lipschitz, it follows from Theorem \ref{th:bigequivalence} that (1) implies (2) when $\ell \circ \gamma^* \in \H$. To show (2) implies (1), let $\Phi = \{f \in \mathrm{Lip}_1(\Omega) : f(x_0)=0\}$, and note that $W_1(P,Q) = \sup_{f \in \Phi}(\E^Qf - \E^Pf)$ by Kantorovich duality. Then (2) implies
\[
\sup_{f \in \Phi}\int_\Omega\ell(\gamma^*(c|f(x)|)P(dx) \le \int_\Omega\ell(\gamma^*(cd(x,x_0)))P(dx) < \infty.
\]
Since $\ell \circ \gamma^*$ grows faster than $x^2$ at infinity, this implies also that $\sup_{f \in \Phi}\E^Pf < \infty$. This in turn implies
\[
\sup_{f \in \Phi}\int_\Omega\ell(\gamma^*(c|f(x) - \E^Pf|)P(dx)
\]
Hence, when $(\ell,\gamma) \in \mathrm{L}\Gamma$, (1) follows from Theorem \ref{th:integralcriterion}.
\end{proof}

\subsection{Applications to options} \label{se:options}
Here we apply Corollary \ref{co:t1-shortfall} to a problem of liquidity risk for options. Let $\Omega = \R_+^n$ represent the space of possible payoffs of $n$ assets at some given maturity. The probability measure $P$ on $\Omega$ represents the joint distribution of payoffs. Let $C^i_k(x) = (x_i-k)^+$ and $P^i_k(x) = (k-x_i)^+$ denote the payoffs of call and put options struck at $k > 0$ on the underlying asset $i$. These are of course the most basic of options, and we will show that their liquidity risk profiles control the liquidity risk profiles of \emph{all Lipschitz options}, i.e. Lipschitz functions of the $n$ underlying assets. Lipschitz functions cover a wide range of option payoffs and investment strategies, such as baskets, straddles, spreads, etc., and a nice feature of this result is that, in a sense, it holds almost independently choice of shortfall risk measure.

\begin{proposition} \label{pr:options}
Suppose $P \in \P^1(\R^n_+)$. Let $\ell$ be a shortfall function, and let $\rho$ be the corresponding shortfall risk measure on $L^1=L^1(\R^n_+,P)$. Let $\gamma$ be a shape function. Assume $(\ell,\gamma) \in \mathrm{L}\Gamma$ and $\ell \circ \gamma^* \in \H$ (see Definitions \ref{def:LGamma} and \ref{def:classH}). The following are equivalent:
\begin{enumerate}
\item There exist $c,k > 0$ such that
\begin{align}
\max_{i=1,\ldots,n}\rho(\lambda(C^i_k - \E C^i_k)) \vee \rho(\lambda(P^i_k - \E P^i_k)) \le \gamma(c\lambda), \ \forall \lambda \ge 0. \label{def:optionmax}
\end{align}
\item There exists $c > 0$ such that for every $f \in \mathrm{Lip}_1(\R^n_+)$ we have
\begin{align}
\rho(\lambda(f - \E f)) \le \gamma(c\lambda), \ \forall \lambda \ge 0. \label{def:optionbound}
\end{align}
\item There exists $c > 0$ such that $\int_{\R^n_+}\ell(\gamma^*(c|x|))P(dx) < \infty$
\end{enumerate}
\end{proposition}
\begin{proof}
By Theorem \ref{th:bigequivalence}, \eqref{def:optionmax} is equivalent to the existence of $c > 0$ such that
\begin{align}
\max_{i=1,\ldots,n}\E[\ell(\gamma^*(c C^i_k))] \vee \E[\ell(\gamma^*(c P^i_k))] < \infty. \label{pf:optionbound1}
\end{align}
Note that $C^i_k(x) + P^i_k(x) = |x_i-k|$ for $x \in \R^n$. For $x\in\R^n$ and for $\vec{k}$ denoting the vector $(k,\ldots,k) \in \R^n$, we have
\begin{align*}
|x - \vec{k}| &= \left[\sum_{i=1}^n|x_i-k|^2\right]^{1/2} \le \left[2\sum_{i=1}^nC^i_k(x)^2 + P^i_k(x)^2\right]^{1/2} \\
	&\le \sqrt{2}\sum_{i=1}^nC^i_k(x) + P^i_k(x) = \sqrt{2}\sum_{i=1}^n|x_i-k| \\
	&\le \sqrt{2n}|x-\vec{k}|.
\end{align*}
Using convexity of $\ell \circ \gamma^*$, we then check easily that \eqref{pf:optionbound1} is equivalent to the existence of $c > 0$ such that
\[
\int_{\R^n_+}\ell(\gamma^*(c_4|x-\vec{k}|))P(dx) < \infty.
\]
Using convexity of $\ell \circ \gamma^*$ once again, this is easily seen to be equivalent to (3).  Finally, Corollary \ref{co:t1-shortfall} implies that (2) and (3) are equivalent.
\end{proof}

\begin{remark}
We would like to be able to set $\gamma(\lambda)$ equal to the left-hand side of \eqref{def:optionmax}, to make \eqref{def:optionbound} as sharp as possible. We cannot do this, for the purely technical reason that we cannot verify in general that $(\ell,\gamma) \in \mathrm{L}\Gamma$ and $\ell \circ \gamma^* \in \H$. 
\end{remark}

Lastly, let us note that two simpler arguments can yield weaker forms of the equivalence of (1) and (2) in Proposition \ref{pr:options}. Fix a $1$-Lipschitz function $f$ on $\R^n_+$, and find $y \in \R^n_+$ such that $f(y) = \E f$, which is possible since the range of $f$ is connected (and thus convex). Then $f(x) - \E f = f(x) - f(y) \le |x-y| \le \sqrt{2}(C_y(x) + P_y(x))$ by the above argument, where we define
\[
C_y(x) = \sum_{i=1}^nC^i_{y_i}(x) = \sum_{i=1}^n(x_i-y_i)^+,
\]
and define $P_y$ similarly. Thus
\[
\rho(\lambda(f-\E f)) \le \rho(\lambda (C_y + P_y)), \ \forall \lambda \ge 0.
\]
Since $C_y \ge 0$ and $P_y \ge 0$, this always exceeds $\rho(\lambda (C_y +P_y- \E C_y- \E P_y))$. More importantly, the choice of strikes $y_1,\ldots,y_n$ depends here on the choice of $f$, whereas the strike $k$ of Proposition \ref{pr:options} was more or less arbitrary. For second attempt, note that
\[
f(x) - \E f = \int_{\R_+^n}(f(x)-f(y))P(dy) \le \int_{\R_+^n}(C_y(x) + P_y(x))P(dy).
\]
Formally using convexity of $\rho$, we conclude that
\[
\rho(\lambda(f-\E f)) \le \int_{\R_+^n}\rho(\lambda (C_y + P_y))P(dy), \ \forall \lambda \ge 0.
\]
Now we're averaging over a range of strikes, and if $\sup_{y \in \R_+^n}\rho(\lambda (C_y + P_y))$ is attained by some $y^* \in \R_+^n$ then we can bound this further by $\rho(\lambda(C_{y^*} + P_{y^*}))$. But this is the \emph{worst} case strike and suffers both of the same problems of the previous argument.

\subsection{Examples} \label{se:examples1}
Let us briefly illustrate how to apply Theorem \ref{th:bigequivalence} to compute some concentration bounds. There is an interesting tradeoff between $\ell$ and $\gamma$; points (3) and (4) of Theorem \ref{th:bigequivalence} depend on $\ell$ and $\gamma$ only through $\ell \circ \gamma^*$. Given two loss functions $\ell_1$ and $\ell_2$ with $\ell_2$ growing faster at infinity than $\ell_1$, the moment condition (4) of Theorem \ref{th:bigequivalence} is a stronger constraint for $\ell_2$, and accordingly the concentration inequality $\rho(\lambda X) \le \gamma(c\lambda)$ is less easily satisfied when $\rho$ corresponds to $\ell_2$ than when it corresponds to $\ell_1$. However, we may \emph{decrease} $\gamma^*$ (or equivalently increase $\gamma$) to compensate for increasing $\ell$, where the terms ``increase'' and ``decrease'' refer loosely to their growth rate at infinity. This tradeoff between $\ell$ and $\gamma$ will be exploited again in Corollary \ref{co:tpequivalence}.

\subsubsection{Subgaussian random variables}
A well known case is when $\rho(X) = \log\E[e^X]$ is the entropic risk measure. The concentration inequality $\rho(\lambda X) \le c\lambda^2/2$ is so pervasive it has been given a name; such an $X$ is called \emph{subgaussian}. The equivalent condition $\E[\exp(c(X^+)^2)] < \infty$ is a well known particular case of Theorem \ref{th:bigequivalence}. Well known subgaussian random variables include bounded random variables (by Hoeffding's lemma) and, of course, normal random variables themselves.

\subsubsection{Moment bounds}
Consider the family
\[
\Phi = \left\{X \in L^1 : \E[(X^+)^{2p}] < \infty\right\},
\]
where $p \ge 2$ is fixed.
Set $\gamma(x) = x^2/2$ and $\ell(x) = [(1+x)^+]^p$, and note that $\gamma^*=\gamma$ on $[0,\infty)$. By construction, $\sup_{X \in \Phi}\E[\ell(\gamma^*(X^+))] < \infty$. We easily check that $(\ell,\gamma) \in \mathrm{L}\Gamma$. Thus, if $\rho$ is the shortfall risk measure corresponding to $\ell$, we may find $c > 0$ such that 
\begin{align}
\rho(\lambda X) \le c^2\lambda^2/2, \text{ for all } \lambda \ge 0 \text{ and } X \in \Phi. \label{ex:moments}
\end{align}
The lognormal law is a common choice for an asset price because of its positivity, so let us see what sorts of concentration inequalities it satisfies. Suppose $X = e^Z$, where $Z$ is a standard normal. Since $X$ has finite moments of all orders, it satisfies $\rho(\lambda X) \le c^2\lambda^2/2$ for the above choice of $\ell(x) = [(1+x)^+]^p$. A sharper choice is $\ell(x) = \exp(a|\log(x)|^2)$ for some $a < 1/2$, for $x > 2$, say, and $\ell$ is extended to remain convex and nondecreasing and to achieve $\ell(0)=1$. Indeed, then $\E[\ell(c|X|^2)]$ is controlled in terms of $\E[\exp(a|Z|^2)]$, which is well known to be finite.

\subsubsection{Bounded random variables}
Suppose $\rho$ is a shortfall risk measure corresponding to some loss function $\ell$.
Note that if $\gamma_c(\lambda) := c\lambda$, then the conjugate is 
\[
\gamma_c^*(x) = \begin{cases}
0 &\text{for } x \le c \\
\infty &\text{for } x > c.
\end{cases}
\]
Hence, if $\rho(\lambda (X) \le c\lambda$ for all $\lambda \ge 0$ and for some $c > 0$, then Proposition \ref{pr:shortfall-tail} implies
\[
P(X > t) \le 1/\ell(\gamma_c^*(t)), \text{ for all } t > 0,
\]
which in turn implies $X \le c$ a.s. On the other hand, if $X \le c$ a.s. then monotonicity of $\rho$ implies $\rho(\lambda X) \le c\lambda$ for all $\lambda \ge 0$. Hence, $\rho(\lambda X) \le c\lambda$ for all $\lambda \ge 0$ if and only if $X \le c$ a.s. But it is worth noting that for bounded random variables a linear shape function typically will not be sharp for small values of $\lambda$. For example, when $\rho(X) = \log\E[e^X]$ is the entropic risk measure, and when $a\le X \le b$ a.s. with $\E X = 0$, Hoeffding's inequality yields $\rho(\lambda X) \le (b-a)^2\lambda^2/8$.

\section{Tensorization and time consistency} \label{se:tensorization}

We now investigate what can be reasonably called the \emph{tensorization} of concentration inequalities. Supposing that two random variables $X_1$ and $X_2$ satisfy $\rho(\lambda X_i) \le \gamma_i(\lambda)$ for all $\lambda \ge 0$ and $i=1,2$, what can be said about the liquidity risk profile of combinations of $X_1$ and $X_2$, in the form $f(X_1,X_2)$? 
Ultimately, we find that it is rather difficult to do any tensorization with risk measures other than the entropic one, or modest perturbations thereof. On the one hand, we will see that this makes the the entropic risk measure almost uniquely convenient for liquidity risk analysis. On the other hand, the analysis of this section will also give some hints for how to sharpen the mathematical results; see Remark \ref{re:sharpen}.

The rest of the section works exclusively with \emph{law invariant risk measures} defined on $L^1(\Omega,\F,P)$, and we assume henceforth that the underlying probability space is nonatomic and standard Borel. Recalling that $\P^1(\R)$ is the subset of $\P(\R)$ consisting of measures with finite first moments, define $\tilde{\rho} : \P^1(\R) \rightarrow (-\infty,\infty]$ by $\tilde{\rho}(P \circ X^{-1}) = \rho(X)$, which is possible thanks to law-invariance. We may then define, for any $\sigma$-field $\G \subset \F$ in $\Omega$ and any $X \in L^1$,
\[
\rho(X | \G) := \tilde{\rho}(P(X \in \cdot | \G)).
\]
This is a $\G$-measurable random variable, defined uniquely up to a.s. equality. Similary, for a random variable $Y$, write $\rho(X | Y) := \rho(X | \sigma(Y))$. Note that if $Y$ is $\G$-measurable then
\[
\rho(X +Y | \G) = \rho(X | \G) + Y, \ a.s.,
\]
for any random variable $X$. If $X$ and $Y$ are independent, then it is straightforward to check that
\[
\rho(f(X,Y) | X) = \rho(f(x,Y))|_{x = X}.
\]
The key definition is the following, and most of this section is devoted to its elaboration:

\begin{definition}
We say that a law-invariant risk measure $\rho$ is \emph{acceptance consistent} if
\[
\rho(X) \le \rho(\rho(X | \G)),
\]
for every sub-$\sigma$-fields $\G \subset \F$ and every $X \in L^1$ such that $\rho(X | \G) \in L^1$. If the inequality is reversed, we say $\rho$ is rejection consistent. If $\rho$ is both acceptance consistent and rejection consistent, we say $\rho$ is \emph{time consistent}.
\end{definition}

See \cite{lacker-lawinvariant} for a thorough discussion of this property as well as several alternative characterizations.
It follows from the results of Kupper and Schachermayer \cite{kupper-schachermayer} (see Corollary 4.7 of \cite{lacker-lawinvariant}) that the only time consistent law invariant risk measure (up to dilations, i.e., replacing $\rho(X)$ by $t^{-1}\rho(tX)$) is the entropic risk measure.
If we seek merely acceptance consistency,
it is shown in  \cite{lacker-lawinvariant} that a shortfall risk measure is acceptance consistent if the loss function is \emph{log-subadditive} in the sense that $\ell(x+y)\le\ell(x)\ell(y)$ for all $x,y \in \R$, and it seems that this sufficient condition is not far from necessary. Aside from the exponential function $\ell(x)=e^x$ and close relatives thereof, not many functions satisfy both this property and those required in the definition of a loss function. Similarly, an optimized certainty equivalent is acceptance consistent if the function $\phi$ in its definition satisfies $x\phi^*(y) + y\phi^*(x) \le \phi^*(xy)$ for $x,y \ge 0$. This holds when $\phi^*(x) = x\log x$ on $x \ge 0$ and $\phi^* = \infty$ on $(-\infty,0)$ or equivalently $\phi(x) = e^{x-1}$. Alternatively, if $\phi^*(x) = x\ell^{-1}(x)$ for a loss function $x$, this property of $\phi^*$ is essentially equivalent to log-subadditivity of $\ell$, indicating that this property of $\phi^*$ is no easier to come by.

While not many risk measures beyond the entropic one are acceptance consistent, the results of this section serve as an illustration of why acceptance consistence is so useful in the study of liquidity risk. Indeed, its relevance is seen first in the following simple but useful result:

\begin{proposition} \label{pr:acc-subadditive}
Let $\rho$ be a acceptance consistent law-invariant risk measure. Suppose $X_1,X_2 \in L^1$ are any independent (real-valued) random variables. Then
\[
\rho(X_1 + X_2) \le \rho(X_1) + \rho(X_2).
\]
Similarly, if $\rho$ is rejection consistent, then $\rho(X_1+X_2) \ge \rho(X_1)+\rho(X_2)$.
\end{proposition}
\begin{proof}
Note that
\[
\rho(X_1 + X_2 | X_1) = \rho(x_1 +  X_2)_{x_1=X_1} = X_1 + \rho(X_2).
\]
Since this is clearly in $L^1$, acceptance consistency implies
\begin{align*}
\rho(X_1 + X_2) &\le \rho\left(\rho(X_1 + X_2 | X_1)\right) \\
	&= \rho\left(X_1+\rho(X_2)\right) \\
	&= \rho(X_1) + \rho(X_2).
\end{align*}
The second claim is proven similarly.
\end{proof}

Note that this may of course be applied inductively. Suppose $X_1,\ldots,X_n$ are independent random variables.
If $\rho(\lambda X_i) \le \gamma_i(\lambda)$ for some shape functions $\gamma_i$, and if $\rho$ is acceptance consistent, then
\[
\rho\left(\lambda\sum_{i=1}^nX_i\right) \le \sum_{i=1}^n\rho(\lambda X_i) \le \sum_{i=1}^n\gamma_i(\lambda).
\]
In other words, if $\overline{X}_n = \frac{1}{n}\sum_{i=1}^nX_i$, then
\[
\rho(\lambda \overline{X}_n) \le \sum_{i=1}^n\gamma_i(\lambda/n).
\]
For example, suppose $\gamma_i = \gamma$ does not depend on $i$ and also that $\gamma(0)=0$. Convexity of $\gamma$ then implies $\gamma(\lambda/n) \le \gamma(\lambda)/n$, and so
\[
\rho(\lambda \overline{X}_n) \le \gamma(\lambda).
\]
This bound is dimension-free, in the sense that the right-hand side does not depend on $n$. In other words, when the risk measure is acceptance consistent, averaging independent losses does not worsen the liquidity risk profile.
But more is true:

\begin{proposition} \label{pr:general-acc-tensorization}
Let $\rho$ be an acceptance consistent law-invariant risk measure. For each $i=1,\ldots,n$ let $X_i$ be an $E_i$-valued random variable, where $(E_i,d_i)$ is a complete separable metric space. Suppose $X_1,\ldots,X_n$ are independent, with respective laws $\mu_i \in \P^1(E_i)$.
Suppose that for each $i$, each $\lambda \ge 0$, and each $1$-Lipschitz function $f$ on $(E_i,d_i)$, we have
\begin{align}
\rho(\lambda(f(X_i)-\E f(X_i)))\le \gamma_i(\lambda). \label{hyp:general-acc-tensorization}
\end{align}
Equip $E = E_1 \times \cdots \times E_n$ with the $\ell^1$-metric
\[
d(x,y) = \sum_{i=1}^nd_i(x_i,y_i).
\]
Then, for any $1$-Lipschitz function $f$ on $E$, we have
\[
\rho\left[\lambda\left(f(X_1,\ldots,X_n) - \E f(X_1,\ldots,X_n)\right)\right] \le \sum_{i=1}^n\gamma_i(\lambda), \text{ for all } \lambda \ge 0.
\]
\end{proposition}
\begin{proof}
We prove the case of $n=2$, as the general case follows by an easy induction.
For fixed $x$, the function $y \mapsto f(x,y)$ is Lipschitz, and thus \eqref{hyp:general-acc-tensorization} implies
\[
\rho(\lambda f(x,X_2)) \le \gamma_2(\lambda) + \lambda\E f(x,X_2).
\]
On the other hand, thanks to Proposition \ref{pr:convexorder}, $\rho(\lambda f(x,X_2)) \ge \lambda \E[f(x,X_2)]$. Combining these upper and lower bounds shows that $x \mapsto \rho(\lambda f(x,X_2))$ is $\mu_1$-integrable, or equivalently $\rho(\lambda f(X_1,X_2) | X_1)$ is in $L^1$.
The function $x \mapsto \E f(x,X_2)$ is $1$-Lipschitz, and thus acceptance consistency and monotonicity yield
\begin{align*}
\rho(\lambda f(X_1,X_2)) &\le \rho(\rho(\lambda f(x,X_2))|_{x = X_1}) \le \gamma_2(\lambda) + \rho(\lambda\E[f(X_1,X_2) | X_1]) \\
	&\le \gamma_1(\lambda) + \gamma_2(\lambda) + \lambda\E[f(X,Y)].
\end{align*}
\end{proof}

\begin{remark}
In Corollary \ref{co:t1}, we will see that the hypothesis \eqref{hyp:general-acc-tensorization} is equivalent to the \emph{transport inequality}
\[
\gamma^*_i(W_1(\nu,\mu_i)) \le \alpha(\nu | \mu_i), \text{ for all } \nu \in \P_{\mu_i}^\infty(E_i),
\]
where $W_1$ is the Wasserstein distance defined in \eqref{def:wasserstein}, and
where $\alpha(\cdot | \mu_i)$ is defined as the minimal penalty function of the risk measure on $L^1(E_i,\mu_i)$ given by $f \mapsto \rho(f(X_i))$. See \cite{lacker-lawinvariant} for more discussion of this functional $\alpha(\cdot | \cdot)$; in particular, the acceptance consistency of $\rho$ is equivalent to this functional $\alpha$ satisfying a property called \emph{superadditivity}, which is simply an inequality form of the chain rule for relative entropy. In fact, Proposition \ref{pr:general-acc-tensorization} could be alternatively proven using this superadditivity property along with now-classical arguments for tensorization of transport inequalities originating in \cite{marton-bounding,talagrand-transportation}. Following the ideas of \cite[Proposition 1.8]{gozlan-leonard-survey}, we could extend this to various classes of non-Lipschitz functions, which corresponds on the dual side to choosing a different optimal transport cost in place of $W_1$ (see Corollary \ref{co:tc}), but we omit this for the sake of brevity.
\end{remark}

We may extend these arguments beyond the independent case, as illustrated by the following analog of Proposition \ref{pr:acc-subadditive}:

\begin{proposition} \label{pr:tensorization-dependent}
Let $\rho$ be an acceptance consistent law-invariant risk measure. Suppose $X_1$ and $X_2$ are any (not necessarily independent) random variables satisfying
\begin{align}
\rho(\lambda X_1) &\le \gamma_1(\lambda), \label{hyp:tensorization-dependent1} \\
\rho(\lambda X_2 | X_1) &\le \gamma_2(\lambda), \ a.s., \label{hyp:tensorization-dependent2}
\end{align}
for all $\lambda \ge 0$, where $\gamma_1$ and $\gamma_2$ are shape functions.
Then
\[
\rho(\lambda(X_1 + X_2)) \le \gamma_1(\lambda) + \gamma_2(\lambda).
\]
\end{proposition}
\begin{proof}
Use the same calculation as in the proof of Proposition \ref{pr:acc-subadditive} along with monotonicity of $\rho$:
\begin{align*}
\rho(\lambda(X_1 + X_2)) &\le \rho\left(\rho(\lambda(X_1 + X_2) | X_1)\right) \\
	&= \rho\left(\lambda X_1+\rho(\lambda X_2 | X_1)\right) \\
	&= \rho(\lambda X_1 + \gamma_2(\lambda)) \\
	&\le \rho(\lambda X_1) + \gamma_2(\lambda) \\
	&\le \gamma_1(\lambda) + \gamma_2(\lambda).
\end{align*}
\end{proof}

Proposition \ref{pr:tensorization-dependent} could naturally apply to Markov chains, by interpreting the two hypotheses \eqref{hyp:tensorization-dependent1} and \eqref{hyp:tensorization-dependent2} as conditions on the initial distribution and transition kernel, respectively. More generally, a natural analog of Proposition \ref{pr:general-acc-tensorization} is available in this setting with dependence, but we will not write this out.

\begin{remark} \label{re:sharpen}
Let us note why the inequalities of this section are generally not very sharp beyond the entropic case, $\rho(X) = \log\E[e^X]$. Only for the entropic risk measure (the essentially unique time consistent risk measure) does the conclusion of Proposition  \ref{pr:acc-subadditive} hold with equality.
Generally, when $\rho$ is acceptance consistent, are losing something at each step of the induction in Proposition \ref{pr:general-acc-tensorization}, namely with the inequality
\[
\rho(\lambda f(X_1,X_2)) \le \rho(\rho(\lambda f(x,X_2))|_{x = X_1}).
\]
Our inequalities remain tight if we instead work directly with the right-hand side above, i.e., with iterations of $\rho$. This idea will be explored in a follow-up paper.
\end{remark}

\section{Dual inequalities and concentration of measure} \label{se:dualinequalities}

We now turn toward a more general discussion of concentration inequalities and their \emph{duals}, of the type described in the following simple proposition, building on the idea of Bobkov and G\"otze \cite{bobkov-gotze}. In fact, in its present form, this is essentially Theorem 3.5 of \cite{gozlan-leonard-survey}.

\begin{proposition} \label{pr:dualinequality}
Let $\rho$ be a risk measure on $L^1$ with a penalty function $\alpha$. Let $\gamma$ be a shape function, and let $X \in L^1$.
Then
\[
\rho(\lambda X) \le \gamma(\lambda), \text{ for all } \lambda \ge 0,
\]
if and only if 
\[
\gamma^*(\E^Q[X])  \le \alpha(Q), \text{ for all } Q \in \P_P^\infty(\Omega).
\]
\end{proposition}
\begin{proof}
Recall that $\gamma^*(t) = \sup_{\lambda \ge 0}(\lambda t - \gamma(\lambda))$, since $\gamma = \infty$ on $(-\infty,0)$ by convention.
Thus $\gamma^*(\E^Q[X])  \le \alpha(Q)$ for all $Q \in \P_P^\infty(\Omega)$ if and only if
\[
\sup_{\lambda \ge 0}\left(\lambda\E^Q[X] - \gamma(\lambda)\right) \le \alpha(Q), \text{ for all } \in \P_P^\infty(\Omega),
\]
which holds if and only if
\[
\sup_{Q \in \P_P^\infty(\Omega)}\left(\lambda\E^Q[X] - \alpha(Q)\right) \le \gamma(\lambda), \text{ for all } \lambda \ge 0.
\]
Recognize that the left-hand side is equal to $\rho(\lambda X)$.
\end{proof}

More generally, we extend this to cover general starting positions as follows. Suppose we already face a loss of $Y$, and we wish to understand the risk of investing some additional quantity into a loss of $X$. Then the relevant liquidity risk profile is that of $X$ relative to $Y$, i.e., $(\rho(\lambda X + Y) -\rho(Y))_{\lambda \ge 0}$. More generally, we consider a class of starting positions $Y$ instead of a single one. It is quite a lot to ask of a random variable $X$ to verify the concentration inequality
\begin{align}
\sup_{Y \in \Phi}\rho(\lambda X + Y) -\rho(Y) \le \gamma(\lambda), \ \forall \lambda \ge 0, \label{def:concentration-startingposition}
\end{align}
when the class $\Phi \subset L^1$ is large. For example, as we will see in Corollary \ref{co:sensitiveinitialposition}, if $\gamma'(0)=\gamma(0)=0$ and $\Phi = L^1$, then $X$ can  satisfy \eqref{def:concentration-startingposition} only if $X \le 0$ a.s.

\begin{proposition} \label{pr:generalizeddualinequality}
Let $\rho$ be a risk measure on $L^1$ with a penalty function $\alpha$. Let $\Phi \subset L^1$.
Then the functional 
\begin{align}
\bar{\rho}(X) := \sup_{Y \in \Phi}\left(\rho(X+Y)-\rho(Y)\right) \label{pr:generalizeddualinequality1}
\end{align}
is a risk measure on $L^1$ with a penalty function given by
\[
\bar{\alpha}(Q) := \alpha(Q) - \sup_{Y \in \Phi}\left(\E^Q[Y] - \rho(Y)\right).
\]
If $\alpha$ is the minimal penalty function for $\rho$, then $\bar{\alpha}$ is the minimal penalty function for $\bar{\rho}$.
\end{proposition}
\begin{proof}
It is clear that \eqref{pr:generalizeddualinequality1} defines a risk measure. 
To see that $\bar{\alpha}$ is a penalty function for $\bar{\rho}$:
\begin{align*}
\sup_{Q \in \P_P^\infty(\Omega)}\left(\E^Q[X] - \bar{\alpha}(X)\right) &= \sup_{Q \in \P_P^\infty(\Omega)}\left[\E^Q[X] - \alpha(Q) + \sup_{Y \in \Phi}\left(\E^Q[Y] - \rho(Y)\right)\right] \\
	&= \sup_{Y \in \Phi}\left[\sup_{Q \in \P_P^\infty(\Omega)}\left(\E^Q[X+Y] - \alpha(Q) \right) -\rho(Y)\right] \\
	&= \sup_{Y \in \Phi}\left(\rho(X+Y)-\rho(Y)\right).
\end{align*}
If $\alpha$ is the \emph{minimal} penalty function for $\rho$, then by Theorem \ref{th:follmerschied} it suffices to verify the formula
\[
\bar{\alpha}(Q) = \sup_{X \in L^1}\left(\E^Q[X] - \bar{\rho}(X)\right).
\]
But this is straightforward.
\end{proof}

Combining Propositions \ref{pr:dualinequality} and \ref{pr:generalizeddualinequality} yields:

\begin{corollary} \label{co:generalizeddualinequality}
Let $\rho$ be a risk measure on $L^1$ with a penalty function $\alpha$. For $\Phi \subset L^1$, the following are equivalent:
\begin{enumerate}
\item $\rho(\lambda X + Y) - \rho(Y) \le \gamma(\lambda)$, for all $\lambda > 0$ and $Y \in \Phi$.
\item 
\[
\gamma^*(\E^Q[X])  \le \alpha(Q) - \sup_{Y \in \Phi}\left(\E^Q[Y] - \rho(Y)\right), \text{ for all } Q \in \P_P^\infty(\Omega).
\]
\end{enumerate}
\end{corollary}

\begin{corollary} \label{co:sensitiveinitialposition}
Let $\rho$ be a risk measure on $L^1$ which admits a penalty function (or equivalently has the Fatou property). Let $\gamma$ be a shape function satisfying $\gamma(0)=0$. Then $X \in L^1$ satisfies
\[
\rho(\lambda X + Y) - \rho(Y) \le \gamma(\lambda), \text{ for all } \lambda > 0, \ Y \in L^1,
\]
if and only if $X \le \gamma'(0)$ a.s., where $\gamma'(0)$ is the right-derivative.
\end{corollary}
\begin{proof}
Let $\alpha$ be the minimal penalty function of $\rho$. Corollary \ref{co:generalizeddualinequality} with $\Phi = L^1$ implies that for all $Q \in \P_P^\infty(\Omega)$
\[
\gamma^*(\E^Q[X])  \le \alpha(Q) - \sup_{Y \in L^1}\left(\E^Q[Y] - \rho(Y)\right) = 0.
\]
Since $\gamma \ge 0$ and $\gamma(0)=0$, if $x \in \R$ satisfies $\gamma^*(x) = 0$ then necessarily $tx \le \gamma(t)$ for all $t > 0$, which implies $x \le \inf_{t > 0}\gamma(t)/t = \gamma'(0)$.
Thus $\E^Q[X] \le \gamma'(0)$ for all $Q \in \P_P^\infty(\Omega)$, which implies $X \le \gamma'(0)$ a.s. Conversely, if $X \le \gamma'(0)$ a.s., then cash additivity and monotonicity yield $\rho(\lambda X + Y)-\rho(Y) \le \lambda \gamma'(0) \le \gamma(\lambda)$, where the last inequality follows from convexity of $\gamma$.
\end{proof}

\subsection{Transport inequalities} \label{se:transportationinequalities}

Proposition \ref{pr:dualinequality} helps us understand how to control the concentration of many random variables uniformly. One special but important case of this is the class of \emph{transport  inequalities}, which we discuss for two reasons: to provide an alternative perspective on concentration inequalities, and to illustrate how our study of shortfall risk measures yields an interesting new characterization of some classical transport inequalities. 
Suppose that $\Omega$ is equipped with a metric $d$ and that $(\Omega,d)$ is a complete, separable metric space.\footnote{More generally, we may assume $\Omega$ is a Polish space and $d$ is a measurable metric on $\Omega$, not necessarily compatible with the topology.} Define the Wasserstein distances $W_p$ for $p \ge 1$ by
\begin{align}
W_p(P,Q) := \left[\inf_{\pi \in \Pi(P,Q)}\int_{\Omega \times \Omega}d(x,y)^p\pi(dx,dy)\right]^{1/p}, \label{def:wasserstein}
\end{align}
where $\Pi(P,Q)$ is the set of probability measures on $\Omega \times \Omega$ with first and second marginals equal to $P$ and $Q$, respectively. Of course, $W_p(P,Q)$ is defined only for $P,Q \in \P^p(\Omega)$, where for an arbitrary fixed $x_0 \in \Omega$ we define
\[
\P^p(\Omega) := \left\{P \in \P(\Omega) : \int_\Omega d(x,x_0)^pP(dx) < \infty \right\}.
\]
The relative entropy between two measures $P,Q \in \P(\Omega)$ is defined as usual by
\[
H(Q | P) = \begin{cases}
\int \log\frac{dQ}{dP}dQ &\text{if } Q \ll P, \\
\infty &\text{otherwise}
\end{cases}.
\]
Classically, one says that the measure $P$ verifies the transport inequality $\mathbf{T}_p$ if there exists $c > 0$ such that
\[
W_p(P,Q) \le \sqrt{cH(Q | P)}, \text{ for all } Q \in \P^p(\Omega).
\]
These inequalities have been thoroughly studied in connection with concentration of measure, as we will elaborate on shortly.
The well known Kantorovich dual description of $W_1$ is
\[
W_1(P,Q) = \sup\left\{\int_\Omega\phi\,d(P-Q) : \phi \in \mathrm{Lip}_1(\Omega)\right\},
\]
where $\mathrm{Lip}_1(\Omega)$ denotes the set of $1$-Lipschitz functions from $(\Omega,d)$ to $\R$. With this in mind, we get the following corollary of Proposition \ref{pr:dualinequality}.

\begin{corollary} \label{co:t1}
Suppose $\rho$ is a risk measure with a penalty function $\alpha$, and let $\gamma$ be a shape function.
Suppose $(\Omega,d)$ is a complete separable metric space and $P \in \P^1(\Omega)$.
The following are equivalent:
\begin{enumerate}
\item For each $f \in \mathrm{Lip}_1(\Omega)$ and each $\lambda \in \R$, we have
\begin{align}
\rho\left(\lambda f\right) \le \gamma(|\lambda|) + \lambda\E f. \label{co:transport-basic}
\end{align}
\item For each $Q \in \P_P^\infty(\Omega)$, we have
\[
\gamma^*(W_1(P,Q)) \le \alpha(Q).
\]
\end{enumerate}
\end{corollary}
\begin{proof}
Note that (1) is equivalent to the same statement but with the restriction that $\lambda \ge 0$, since $\mathrm{Lip}_1(\Omega) = -\mathrm{Lip}_1(\Omega)$.
Thanks to the cash additivity of $\rho$,  the inequality \eqref{co:transport-basic} is equivalent to
\[
\rho\left(\lambda \left(f - \E f\right)\right) \le \gamma(\lambda).
\]
According to Proposition \ref{pr:dualinequality}, this is equivalent to
\[
\gamma^*(\E^Qf - \E^P f) \le \alpha(Q), \ \forall Q \in \P_P^\infty(\Omega).
\]
Taking the supremum over $f \in \mathrm{Lip}_1(\Omega)$, the claim follows from the dual formula for $W_1$.
\end{proof}

Let us generalize this discussion somewhat, in the spirit of \cite{gozlan-leonard-largedeviation,gozlan-leonard-survey}, by considering more general transport costs. Suppose $c : \Omega^2 \rightarrow [0,\infty]$ is lower semicontinuous, and assume for normalization purposes that $\inf_yc(x,y)=0$ for all $x$. Let 
\[
\T_c(P,Q) := \inf_{\pi \in \Pi(P,Q)}\int_{\Omega \times \Omega}c(x,y)\pi(dx,dy)
\]
Let $\Phi_c$ denote the set of pairs of bounded measurable functions $(\phi,\psi)$ on $\Omega$ satisfying $\phi(x) + \psi(y) \le c(x,y)$ for all $x,y \in \Omega$.
By Kantorovich duality (see, e.g., \cite{villani-book}), we have
\begin{align*}
\T_c(P,Q) = \sup_{(\phi,\psi) \in \Phi_c}\E^Q[\psi] + \E^P[\phi].
\end{align*}
With this in mind, we arrive at a characterization of more general transport inequalities, as a corollary of Proposition \ref{pr:dualinequality}:

\begin{corollary} \label{co:tc}
Suppose $\rho$ is a risk measure with penalty function $\alpha$, and let $\gamma$ be a shape function.
The following are equivalent:
\begin{enumerate}
\item For each $(\phi,\psi) \in \Phi_c$ and each $\lambda \ge 0$, we have
\begin{align*}
\rho\left(\lambda \psi\right) \le \gamma(\lambda) - \lambda\E^P[ \phi].
\end{align*}
\item For each $Q \in \P_P^\infty(\Omega)$, we have
\[
\gamma^*(T_c(P,Q)) \le \alpha(Q).
\]
\end{enumerate}
\end{corollary}

Recall that in Theorem \ref{th:bigequivalence}, points (3) and (4) depend on $\ell$ and $\gamma$ only through $\ell \circ \gamma^*$. By varying $\ell$ and $\gamma$ while keeping the composition $\ell \circ \gamma^*$ unchanged, the implications surrounding points (1) and (2) in Theorem \ref{th:bigequivalence} are rather interesting. What follows is an application to transport inequalities, characterizing $\textbf{T}_p$ in terms of a seemingly weaker inequalities, up to a change in constant which we do not track carefully.

\begin{corollary} \label{co:tpequivalence}
Fix $p \in \{1,2\}$.
The $\mathbf{T}_p$ inequality
\[
W_p(P,Q) \le c\sqrt{H(Q | P)}, \ \forall Q \in \P_P^\infty(\Omega) \label{co:t1-def}
\]
holds for some $c > 0$ if and only if
\begin{align}
W_p(P,Q) \le c\sqrt{\log\left\|\frac{dQ}{dP}\right\|_{L^\infty(P)}}, \ \forall Q \in \P_P^\infty(\Omega) \label{co:t1log}
\end{align}
for some $c > 0$.
\end{corollary}
\begin{proof}
The first direction is trivial, since if $\nu \ll \mu$ then 
\[
H(\nu | \mu) = \int \log\frac{d\nu}{d\mu}d\nu \le \log\left\|\frac{d\nu}{d\mu}\right\|_{L^\infty(\mu)}.
\]
To prove the converse, set $c(x,y) = |x-y|^p$ so that $\T_c = W_p^p$. Set
\[
\Phi = \{\psi + \E^P[\phi] : (\phi,\psi) \in \Phi_c \}.
\]
Set $\gamma_1^*(x) = \exp(t^{2/p})-1$. For a new constant $c > 0$ we may rewrite \eqref{co:t1log} as
\[
\gamma_1^*\left(c\sup_{f \in \Phi}\E^Q[f]\right) \le \|dQ/dP\|_{L^\infty(P)} - 1 = \alpha_1(Q), \ \forall Q \in \P_P^\infty(\Omega),
\]
where $\alpha_1$ is the minimal penalty function of the shortfall risk measure associated to the loss function $\ell_1(t) = (1+t)^+$ (see Section \ref{se:tails-shortfall}). Since the function $\ell_1 \circ \gamma_1^*(t) = \exp(t^{2/p})$ is in $\H$, we conclude from Theorem \ref{th:bigequivalence} that (for some $c$)
\[
\sup_{f \in \Phi}\E^P\left[\exp\left(c|f^+|^{2/p}\right)\right] = \sup_{f \in \Phi}\E^P\left[\ell_1(\gamma_1^*(cf^+))\right] < \infty.
\]
Now, set $\ell(t) = e^t$ and $\gamma(t) = t^{2/p}$, noting that either $\gamma'(0)=1>0$ (if $p=2$) or $\gamma''(0)=2>0$ (if $p=1$). This pair $(\ell,\gamma)$ is in $\mathrm{L}\Gamma$ and $\ell \circ \gamma^* = \ell_1 \circ \gamma_1^*$, and so Theorem \ref{th:bigequivalence} shows that this implies (again for a new constant $c$)
\[
\gamma^*\left(c\sup_{f \in \Phi}\E^Q[f]\right) \le \alpha(Q), \ \forall Q \in \P_P^\infty(\Omega),
\]
where $\alpha$ is the minimal penalty function of the shortfall risk measure associated to the loss function $\ell$. But then $\alpha(Q) = H(Q | P)$, and we arrive at the desired $\mathbf{T}_p$ inequality.
\end{proof}

\subsection{Concentration functions}
Let us briefly recast the preceding results in a perspective more familiar in the theory of concentration of measure.
Concentration of measure is classically defined in terms of enlargements of sets, and we will explain this briefly, following the presentation of Ledoux \cite{ledoux-concentration}.
Suppose throughout the section that $\Omega$ is a complete separable metric space.
The following function $C_P : (0,\infty) \rightarrow [0,1/2]$ is the \emph{concentration function} of the measure $P$:
\[
C_P(r) := \sup\left\{1 - P(A^r) : A \in \F, \ P(A) \ge 1/2\right\},
\]
where $A^r := \{\omega \in \Omega : d(\omega,A) < r\}$. Note that $C_P$ decreases to $0$ as $r \rightarrow \infty$. Gaussian concentration bounds of the form $C_P(r) \le c_1\exp(-c_2r^2)$ are strikingly ubiquitous, and this well known phenomenon is called \emph{concentration of measure}.

Given a real-valued random variable $f$ defined on $\Omega$, we say $m_f \in \R$ is a median of $f$ if $P(f \ge m_f) \ge 1/2$ and $P(f \le m_f) \ge 1/2$. Let $\text{Lip}_1(\Omega)$ denote the set of all $1$-Lipschitz function on $(\Omega,d)$. It is known \cite[Proposition 1.3]{ledoux-concentration} that
\[
C_P(r) = \sup\left\{P(f \ge m_f + r) : f \in \text{Lip}_1(\Omega), \ m_f \text{ is a median of } f\right\}.
\]
Moreover, by \cite[Proposition 1.7]{ledoux-concentration}, we can replace the median by the mean in the following sense:
\[
C_P(r) \le \sup\left\{P\left(f \ge \E^Pf + r/2\right) : f \in \text{Lip}_1(\Omega)\right\}.
\]
This leads us to a new description of concentration functions in terms of transport inequalities, which follows from  Theorem \ref{th:bigequivalence} and Corollary \ref{co:t1}:

\begin{proposition} \label{pr:concentrationfunction}
Suppose $\gamma$ is a shape function, $\ell$ is a loss function, $E$ is a Polish space, and $\mu \in \P^1(E)$. Let $\alpha$ the minimal penalty function of the shortfall risk measure $\rho$ associated to $\ell$. Consider the following two statements:
\begin{enumerate}
\item There exists $c > 0$ such that $P$ satisfies the transport inequality
\[
\gamma^*(cW_1(Q,P)) \le \inf_{t > 0}\frac{1}{t}\left\{1 + \E^P\left[\ell^*\left(t\frac{dQ}{dP}\right)\right]\right\}, \text{ for } \nu \ll \mu.
\]
\item There exists $c > 0$ such that $C_P(r) \le 1/\ell(\gamma^*(cr))$ for all $r > 0$.
\end{enumerate}
Then (1) implies (2), and the constant in (2) is half of that of (1). If  $(\ell,\gamma) \in \mathrm{L}\Gamma$ with $\ell \circ \gamma^* \in \H$ (see Definitions \ref{def:LGamma} and \ref{def:classH}), then (2) implies (1).
\end{proposition}

An alternative, arguably more direct argument is available for a result similar to the implication $(1) \Rightarrow (2)$ of Proposition \ref{pr:concentrationfunction}, originally due to Marton \cite{marton-blowingup,marton-bounding} and extended by several authors since, as in \cite[Theorem 1.7]{gozlan-leonard-survey}. Suppose $P$ satisfies the transport inequality of Proposition \ref{pr:concentrationfunction}(1).
Fix a measurable set $A \subset \Omega$ with $P(A) \ge 1/2$, and set $B := \Omega \backslash A^r$ where $r > 0$ is fixed. For a measurable set $C$ define the new probability measure $P_C := P(\cdot \cap C)/P(C)$. Formally  interchanging the order of infimum and supremum, we compute
\begin{align}
\alpha(P_A ) &= \inf_{t > 0}\frac{1}{t}\left(1 + P(A)\ell^*(t/P(A))\right) = \sup_{s \in \R}\inf_{t > 0}\frac{1}{t}\left(1 + P(A)[st/P(A) - \ell(s)]\right) \label{def:marton-infsup} \\
	&= \sup\{s : P(A)\ell(s) \le 1\} = \ell^{-1}(1/P(A)) \le \ell^{-1}(2). \nonumber
\end{align}
Similarly,
\begin{align*}
\alpha(P_B) &= \ell^{-1}(1/P(B)) = \ell^{-1}\left(\frac{1}{1 - P(A^r)}\right).
\end{align*}
Since $d(A,B) \ge r$ and any coupling of $P_A$ and $P_B$ is concentrated on $A \times B$, we have $W_1(P_A,P_B) \ge r$. The triangle inequality for $W_1$ yields
\begin{align*}
r &\le W_1(P_A,P_B) \le W_1(P_A,P) + W_1(P_B,P) \\
	&\le (\gamma^*)^{-1}(\alpha(P_A)) + (\gamma^*)^{-1}(\alpha(P_B)) \\
		&\le (\gamma^*)^{-1}(\ell^{-1}(2)) + (\gamma^*)^{-1}\left(\ell^{-1}\left(\frac{1}{1 - P(A^r)}\right)\right).
\end{align*}
It follows that, if $r_0 := (\gamma^*)^{-1}(\ell^{-1}(2))$, then
\[
P(A^r) \ge 1 - \frac{1}{\ell \circ \gamma^*(r - r_0)},
\]
which in turn implies $C_P(r) \le 1/\ell(\gamma^*(r-r_0))$. Of course, the step \eqref{def:marton-infsup} remains to be justified, but it suffices in fact to show that if $a,b > 0$ then $at \le 1 + b\ell^*(t/b)$ for all $t > 0$ if and only if $a \le \ell^{-1}(1/b)$. But this is straightforward, since the former condition is equivalent to
\[
1 \ge \sup_{t > 0}(at - b\ell^*(t/b)) = b\sup_{t > 0}(at/b - \ell^*(t/b)) = b\ell(a).
\]
The last equality holds because $\ell^*(t) = \infty$ for $t < 0$ thanks to nonnegativity of $\ell$.

\subsection{Extensions} \label{se:examples}
In connection with transport inequalities, relative entropy is by far the most well studied of the examples of penalty functions mentioned above. However, the work of Guillin et al. \cite{guillin-leonard-wu-yao} successfully used the Donsker-Varadhan information (also known as the Fisher information) as a substitute for relative entropy in transport inequalities, motivated by rates of convergence of Markov processes. Indeed, the convex conjugate of the Donsker-Varadhan information is a risk measure related to the log moment generating function of a time-integrated Markov process, and this indeed underlies the proofs in \cite{guillin-leonard-wu-yao}. 
In connection with random matrix theory, Biane and Voiculescu \cite{biane-voiculescu} (see also \cite[Section 12]{gozlan-leonard-survey}) showed that the semi-circle law satisfies a quadratic transport inequality involving the so-called ``free entropy.''
The potential applications of dual inequalities (defined in full generality in Proposition \ref{pr:dualinequality}) are as countless as the convex functionals of probability measures (or equivalently, their conjugates: risk measures) appearing in various applications.

To mention one final example: Suppose $\Omega=\R^n$ and a penalty function $\alpha(Q)$ is defined as the \emph{martingale optimal transport cost} from $P$ to $Q$, recently introduced in \cite{beiglbock-henrylabordere-penkner}. The corresponding risk measure $\rho$ can be computed using a form of Kantorovich duality. In particular, let $c : \Omega^2 \rightarrow [0,\infty)$ be lower semicontinuous, and fix $\mu \in \P(\Omega)$ with finite first moment. Define
\[
\rho(f) = -\int_\Omega(c(x,\cdot) - f)^{**}(x)\mu(dx),
\]
where $\phi^{**}$ denotes the biconjugate or convex envelope of $\phi$. It is easily checked that $\rho$ is a risk measure (though not necessarily normalized), and its minimal penalty function is
\begin{align}
\alpha(\nu) = \inf_{\pi \in \Pi_m(\mu,\nu)}\int_{\Omega^2}c\,d\pi, \label{def:mtgkantorovich}
\end{align}
where $\Pi_m(\mu,\nu)$ denotes the (possibly empty) set of martingale couplings, i.e., the set of laws of random vectors $(X,Y)$ where $X \sim \mu$, $Y \sim \nu$, and $\E[Y | X] = X$. Indeed, the duality formula \eqref{def:mtgkantorovich} is studied carefully in several recent papers \cite{beiglboeck-juillet,beiglbock-nutz-touzi}
, with the latter paper \cite[Remark 7.9]{beiglbock-nutz-touzi} discussing the simplified form we presented in \eqref{def:mtgkantorovich}. Proposition \ref{pr:dualinequality} tells us that, if $f \in B(\Omega)$, 
\[
\rho(\lambda f) \le \gamma(\lambda), \text{ for all } \lambda \ge 0,
\]
if and only if 
\[
\gamma^*\left(\int_\Omega f\,d\nu\right) \le \alpha(\nu), \text{ for all } \nu \ll \mu.
\]
This indeed encompasses simplified forms of several martingale inequalities; for example, if $c(x,y) = -(x \vee y)^2$ and $f(y) = -4y^2$, this is nothing but Doob's inequality.
One should of course extend this idea to martingales with a larger time index set, but, conceivably, similar convex-analytic methods would apply. 
We cut the discussion short here, as any further exploration of this idea is beyond the scope of this paper.

\subsection*{Acknowledgements}
The author would like to thank Ramon van Handel, Zach Feinstein, Alexander Schied, and Thaleia Zariphopoulou for helpful discussions and invaluable advice regarding the organization of the paper.

\appendix

\section{Proof of Theorem \ref{th:integralcriterion}} \label{ap:proof}

\begin{lemma} \label{le:condition5}
If $(\ell,\gamma) \in \mathrm{L}\Gamma$, then
\[
\lim_{c \rightarrow \infty}\limsup_{n \rightarrow\infty}\sup_{x \ge c}\frac{\ell(\gamma^*(x/n))}{\ell(\gamma^*(x))} = 0.
\]
\end{lemma}
\begin{proof}
For a fixed $n > 1$ and $c > 0$, we have $\ell(\gamma^*(c/n)) \le \ell(\gamma^*(c))$. But clearly $\ell(\gamma^*(x/n))-\ell(\gamma^*(x)$ is increasing in $x$, and so
\[
\sup_{x \ge c}\frac{\ell(\gamma^*(x/n))}{\ell(\gamma^*(x))} = \frac{\ell(\gamma^*(c/n))}{\ell(\gamma^*(c))}.
\]
Since $\ell$ and $\gamma^*$ are continuous with $\gamma^*(0)=-\gamma(0) \le 0$ and $\ell(0)=1$, we have
\[
\limsup_{n \rightarrow\infty}\sup_{x \ge c}\frac{\ell(\gamma^*(x/n))}{\ell(\gamma^*(x))} = \frac{\ell(-\gamma(0))}{\ell(\gamma^*(c))} \le \frac{1}{\ell(\gamma^*(c))}
\]
for each $c > 0$. Since $\gamma(x) < \infty$ for some $x > 0$, it follows that $\gamma^*(x)$ tends to infinity as $x \rightarrow \infty$. Thus so does $\ell(\gamma^*(x))$ (since $\ell(x) > 1$ for all $x > 0$ by definition of a loss function), and this completes the proof.
\end{proof}

Let us now turn toward the proof of Theorem \ref{th:integralcriterion}. Let us check first that
without loss of generality we may assume $\E X = 0$ for each $X \in \Phi$. Given a general $\Phi$, set $\Phi' := \{X - \E X : X \in \Phi\}$. The integrability condition
\[
\sup_{X \in \Phi}\E[\ell(\gamma^*(\kappa X^+))] < \infty
\]
implies one for $\Phi'$, albeit with a different $\kappa$, since
\[ 
\E[\ell(\gamma^*(\kappa (X - \E X)^+/2))] \le \frac{1}{2}\left\{\E[\ell(\gamma^*(\kappa X^+))] + \E[\ell(\gamma^*(\kappa (- \E X)^+))]\right\}.
\]
Since $\ell \circ \gamma^*$ is nonconstant, we have $\sup_{X \in \Phi}\E|X| < \infty$, and we conclude that 
\[
\sup_{X \in \Phi'}\E[\ell(\gamma^*(\kappa X^+/2))] < \infty.
\]
Hence, in the rest of the proof we assume $\E X = 0$ for all $X \in \Phi$.

Now, to prove Theorem \ref{th:integralcriterion}, 
it suffices to show that there exists $n > 0$ such that $\rho(\lambda X) \le \gamma(n \lambda)$ for all $\lambda \ge 0$, or equivalently $\rho(\lambda X/n - \gamma(\lambda))$ for all $\lambda \ge 0$. This is equivalent to finding $n > 0$ such that
\begin{align}
\sup_{X \in \Phi}\sup_{\lambda \ge 0}\E\left[\ell\left(\frac{\lambda X}{n} - \gamma(\lambda)\right)\right] \le 1. \label{pf:integralcriterion1}
\end{align}
The main idea is that the integral bound \eqref{th:integralcriterion-hypothesis} yields some uniform integrability, in the sense that $\ell(\lambda X/n -\gamma(\lambda)) \le \ell(\gamma^*(X/n))$; for each $X$ and $\lambda$ we can then show that $\E[\ell(\lambda X/n -\gamma(\lambda))] \rightarrow \ell(-\gamma(\lambda))$ as $n\rightarrow \infty$. But $\ell(-\gamma(\lambda)) < \ell(0) = 1$ for all $\lambda > 0$. For fixed $\delta > 0$, we can make this limit uniform over $\lambda \ge \delta$ to conclude that \eqref{pf:integralcriterion1} holds for some $n$ as long as $\lambda \ge \delta$, but small values of $\lambda$ require more care.

\textbf{Case 1:} Suppose condition (3a) of Definition \ref{def:LGamma} holds. Thanks to Lemma \ref{le:condition5}, we may find $c,N > 0$ such that for $n \ge N$
\[
\sup_{x \ge c}\frac{\ell(\gamma^*(x/(\kappa n)))}{\ell(\gamma^*(x))} \le \frac{\epsilon}{M}.
\]
Since $\ell$ and $\gamma^*$ are nondecreasing, for all $\lambda >0$ and  $X \in \Phi$ we have
\[
\ell\left(\frac{\lambda X}{n} - \gamma(\lambda)\right) \le \ell(\gamma^*(X^+/n)).
\]
Thus, using the hypothesis \eqref{th:integralcriterion-hypothesis} ,
\begin{align*}
\E\left[\ell\left(\frac{\lambda X}{n} - \gamma(\lambda)\right)1_{X > c/\kappa}\right] &\le \E\left[\ell(\gamma^*(\kappa X^+))\frac{\ell(\gamma^*(X^+/n))}{\ell(\gamma^*(\kappa X^+))}1_{X >  c/\kappa}\right] \nonumber \\
	&\le \frac{\epsilon}{M}\E\left[\ell(\gamma^*(\kappa X^+))\right] \le \epsilon. 
\end{align*}
On the other hand,
\[
\E\left[\ell\left(\frac{\lambda X}{n} - \gamma(\lambda)\right)1_{|X| \le c/\kappa}\right] \le \ell\left(\frac{\lambda c}{\kappa n} - \gamma(\lambda)\right) \le \ell(\gamma^*(c/\kappa n)).
\]
Thus
\begin{align}
\E\left[\ell\left(\frac{\lambda X}{n} - \gamma(\lambda)\right)\right] \le \epsilon + \ell(\gamma^*(c/\kappa n)). \label{pf:integralcriterion2}
\end{align}
Now note that $\gamma^*(0) = \sup_{t \ge 0}(-\gamma(t)) = -\gamma(0) < 0$. Since $\gamma$ is nonconstant, the domain $\{t \ge 0 : \gamma^*(t) < \infty\}$ has nonempty interior, on which $\gamma^*$ is continuous. We may thus find $t_0 > 0$ such that $\gamma^*(t_0) < 0$. Choose $n > c/(\kappa t_0)$ to get $\gamma^*(c/\kappa n) \le \gamma^*(t_0)$ by monotonicity of $\gamma^*$. Since $\ell'(0)$ exists and is strictly positive, and since $\ell(0)=1$, we know that $\ell(x) < 1$ for all $x < 0$. Finally, choose
\[
\epsilon := [1 - \ell(\gamma^*(t_0))] > 0,
\]
and conclude from \eqref{pf:integralcriterion2} that \eqref{pf:integralcriterion1} holds for sufficiently large $n$.

\textbf{Case 2, Step 1:} Suppose condition (3b) of Definition \ref{def:LGamma} holds. We first take care of small values of $\lambda$. We wish to find $\delta > 0$  and $N > 0$ such that
\begin{align}
\sup_{n \ge N}\sup_{\lambda \le \delta}\sup_{X \in \Phi}\E\left[\ell\left(\frac{\lambda X}{n} - \gamma(\lambda)\right)\right] \le 1.  \label{pf:integralcriterion3}
\end{align}
Note that 
\begin{align*}
\frac{d}{d\lambda}\ell\left(\frac{\lambda X}{n} - \gamma(\lambda)\right) &= \ell'\left(\frac{\lambda X}{n} - \gamma(\lambda)\right)\left(\frac{X}{n} - \gamma'(\lambda)\right).
\end{align*}
By the mean value theorem, for $\lambda > 0$ we may find $t_\lambda \in [0,\lambda]$ such that
\begin{align*}
\E\left[\ell\left(\frac{\lambda X}{n} - \gamma(\lambda)\right)\right] &= 1 + \lambda\E\left[\ell'\left(\frac{t_\lambda X}{n} - \gamma(t_\lambda)\right)\left(\frac{X}{n} - \gamma'(t_\lambda)\right) \right].
\end{align*}
To prove \eqref{pf:integralcriterion3} it now suffices to show
\begin{align}
\limsup_{\delta\downarrow 0}\limsup_{n\rightarrow\infty}\sup_{\lambda \in [0,\delta]}\sup_{X \in \Phi}\E\left[\ell'\left(\frac{t_\lambda X}{n} - \gamma(t_\lambda)\right)\left(\frac{X}{n} - \gamma'(t_\lambda)\right) \right] < 0. \label{pf:integralcriterion4}
\end{align}
Fix $\epsilon > 0$ and $\delta > 0$. By assumption (3b), we may find $x_0 > 0$ such that $x\ell'(x) \le \epsilon\ell(\gamma^*(x))$ for all $x \ge x_0$
Note also that $\ell'$ is nonnegative and nondecreasing, by the definition of a loss function. For each $\lambda \le \delta$ we have
\[
\frac{X}{n}\ell'\left(\frac{t_\lambda X}{n} - \gamma(t_\lambda)\right)1_{X \ge x_0} \le \frac{X^+}{n}\ell'\left(\frac{\delta X^+}{n}\right)1_{X \ge x_0} \le \frac{\epsilon}{\delta}\ell \circ \gamma^*\left(\frac{\delta X^+}{n}\right).
\]
Since $\gamma^*(0)=0$ and $\ell(0)=1$, convexity of $\ell\circ\gamma^*$ implies
\begin{align}
\ell \circ \gamma^*(t x) \le t\ell \circ \gamma^*(x) + (1-t). \label{pf:integralcriterion5}
\end{align}
Thus, whenever $n\kappa \ge \delta$,
\[
\frac{X}{n}\ell'\left(\frac{t_\lambda X}{n} - \gamma(t_\lambda)\right) \le \frac{\epsilon}{\kappa n}\ell( \gamma^*(\kappa X^+)) + \frac{\epsilon}{\delta}\left(1 - \frac{\delta}{\kappa n}\right).
\]
On the other hand,
\[
\frac{X}{n}\ell'\left(\frac{t_\lambda X}{n} - \gamma(t_\lambda)\right)1_{|X| < x_0} \le \frac{x_0}{n}\ell'\left(\frac{t_\lambda x_0}{n} - \gamma(t_\lambda)\right) \le \frac{x_0}{n}\ell'(\delta x_0/n),
\]
which clearly tends to zero as $n\rightarrow\infty$.
We conclude from the last two displays along with the hypothesis \eqref{th:integralcriterion-hypothesis} that
\[
\limsup_{n\rightarrow\infty}\sup_{\lambda \in [0,\delta]}\sup_{X \in \Phi}\E\left[\frac{X}{n}\ell'\left(\frac{t_\lambda X}{n} - \gamma(t_\lambda)\right)\right] \le  \epsilon/\delta.
\]
Since this holds for each $\epsilon > 0$, the limsup is in fact zero, and we get 
\[
\limsup_{\delta\downarrow 0}\limsup_{n\rightarrow\infty}\sup_{\lambda \in [0,\delta]}\sup_{X \in \Phi}\E\left[\ell'\left(\frac{t_\lambda X}{n} - \gamma(t_\lambda)\right)\frac{X}{n}\right] \le 0.
\]
The proof of \eqref{pf:integralcriterion4} will be complete as soon as we check that
\begin{align}
\limsup_{\delta\downarrow 0}\limsup_{n\rightarrow\infty}\sup_{\lambda \in [0,\delta]}\sup_{X \in \Phi}-\gamma'(t_\lambda)\E\left[\ell'\left(\frac{t_\lambda X}{n} - \gamma(t_\lambda)\right)\right] < 0. \label{pf:integralcriterion6}
\end{align}
Note that $\gamma'(t_\lambda) \ge \gamma'(0) > 0$ and also
\[
\ell'\left(\frac{t_\lambda X}{n} - \gamma(t_\lambda)\right) \ge \ell'\left(-\frac{\delta |X|}{n} - \gamma(\delta)\right).
\]
Since $\ell'$ is continuous, nondecreasing, and nonnegative, the random variables on the right-hand side are bounded, uniformly in $X$ and $n$, and increase pointwise to the constant $\ell'(-\gamma(\delta))$ for each $X$. The convergence is uniform by Dini's theorem, and thus
\begin{align}
\lim_{n\rightarrow\infty}\sup_{X \in \Phi}\E\left[\ell'\left(-\frac{\delta |X|}{n} - \gamma(\delta)\right)\right] = \ell'(- \gamma(\delta)). \label{pf:integralcriterion7}
\end{align}
Thus
\[
\limsup_{\delta\downarrow 0}\limsup_{n\rightarrow\infty}\sup_{\lambda \in [0,\delta]}\sup_{X \in \Phi}-\gamma'(t_\lambda)\E\left[\ell'\left(\frac{t_\lambda X}{n} - \gamma(t_\lambda)\right)\right] \le -\gamma'(0)\limsup_{\delta\downarrow 0}\ell'(- \gamma(\delta)) = -\gamma'(0)\ell'(0).
\]
By hypothesis, this quantity is strictly negative, proving \eqref{pf:integralcriterion6}.

\textbf{Case 2, Step 2:} We now turn to large values of $\lambda$, now that $\delta > 0$ has been fixed, and show that there exists $N > 0$ such that
\begin{align}
\sup_{\lambda \ge \delta}\sup_{n \ge N}\E\left[\ell\left(\frac{\lambda X}{n} - \gamma(\lambda)\right)\right] \le 1.  \label{pf:integralcriterion8}
\end{align}
Since $\ell'(0)$ exists and is strictly positive, and since $\ell(0)=1$, we know that $\ell(x) < 1$ for all $x < 0$. Thus $\ell(-\gamma(\delta)) < 1$, and with great foresight define
\[
\epsilon := [1 - \ell(-\gamma(\delta))]/2 > 0.
\]
Thanks to Lemma \ref{le:condition5}, we may find $c,N > 0$ such that for $n \ge N$
\[
\sup_{x \ge c}\frac{\ell(\gamma^*(x/(\kappa n)))}{\ell(\gamma^*(x))} \le \frac{\epsilon}{M}.
\]
Note also that, for all $\lambda >0$,
\[
\ell\left(\frac{\lambda X}{n} - \gamma(\lambda)\right) \le \ell(\gamma^*(|X|/n)).
\]
Thus, using the hypothesis \eqref{th:integralcriterion-hypothesis},
\begin{align}
\E\left[\ell\left(\frac{\lambda X}{n} - \gamma(\lambda)\right)1_{X > c/\kappa}\right] &\le \E\left[\ell(\gamma^*(\kappa X^+))\frac{\ell(\gamma^*(X^+/n))}{\ell(\gamma^*(\kappa X^+))}1_{X >  c/\kappa}\right] \nonumber \\
	&\le \frac{\epsilon}{M}\E\left[\ell(\gamma^*(\kappa X^+))\right] \le \epsilon. \label{pf:integralcriterion9}
\end{align}
Now note that for any $t > 0$, the supremum in $\sup_{\lambda \ge \delta}(\lambda t - \gamma(\lambda))$ is attained 
by $\lambda = \inf\{s \ge \delta : \gamma'(s) \ge t\}$, where $\gamma'$ is the right-derivative of $\gamma$. So if $t \le \gamma'(\delta)$, we have $\sup_{\lambda \ge \delta}\ell(\lambda t - \gamma(\lambda)) = \ell(\delta t - \gamma(\delta))$. Since $\gamma(x) > 0 = \gamma(0)$ for all $x > 0$, and since $\gamma$ is increasing and convex, we have $\gamma'(\delta) > 0$. Hence, if
\[
n \ge \frac{c}{\kappa\gamma'(\delta)},
\]
then
\begin{align}
\sup_{\lambda \ge \delta}\E\left[\ell\left(\frac{\lambda X^+}{n} - \gamma(\lambda)\right)1_{X \le c/\kappa}\right] \le \ell\left(\frac{\delta c}{n} - \gamma(\delta)\right). \label{pf:integralcriterion10}
\end{align}
Since the righthand side converges to $\ell(-\gamma(\delta))$ as $n\rightarrow\infty$, for sufficiently large $n$ (again depending only on $c$ and $\delta$) we have
\[
\ell\left(\frac{\delta c}{n} - \gamma(\delta)\right) \le \ell(-\gamma(\delta)) + \epsilon.
\]
Combining this with \eqref{pf:integralcriterion9} and \eqref{pf:integralcriterion10} shows that for sufficiently large $n$
\[
\sup_{\lambda \ge \delta}\E\left[\ell\left(\frac{\lambda X^+}{n} - \gamma(\lambda)\right)\right] \le 2\epsilon + \ell(-\gamma(\delta)) = 1.
\]
Combining this with \eqref{pf:integralcriterion3} completes the proof of \eqref{pf:integralcriterion1}, and thus finishes Case 2.

\textbf{Case 3, Step 1:} Suppose condition (3b) of Definition \ref{def:LGamma} holds. Again, we first take care of small values of $\lambda$, by finding $\delta > 0$ and $N > 0$ such that \eqref{pf:integralcriterion3} holds.
Note that
\begin{align*}
\frac{d^2}{d\lambda^2}\ell\left(\frac{\lambda X}{n} - \gamma(\lambda)\right) &= \ell''\left(\frac{\lambda X}{n} - \gamma(\lambda)\right)\left[\frac{X}{n} - \gamma'(\lambda)\right]^2 - \gamma''(\lambda)\ell'\left(\frac{\lambda X}{n} - \gamma(\lambda)\right).
\end{align*}
By Taylor's theorem, for $\lambda > 0$ we may find $t_\lambda \in [0,\lambda]$ such that
\begin{align}
\E\left[\ell\left(\frac{\lambda X}{n} - \gamma(\lambda)\right)\right] &= 1 - \lambda\gamma'(0)\ell'(0) + \frac{\lambda^2}{2}\E\left[A(\lambda,X) - B(\lambda,X)\right], \label{pf:taylor}
\end{align}
where
\begin{align*}
A_n(\lambda,X) &:= \E\left[\ell''\left(\frac{t_\lambda X}{n} - \gamma(t_\lambda)\right)\left(\frac{X}{n^2} - \gamma'(t_\lambda)\right)^2\right], \\
B_n(\lambda,X) &:= \gamma''(t_\lambda)\E\left[\ell'\left(\frac{t_\lambda X}{n} - \gamma(t_\lambda)\right)\right].
\end{align*}
and where we used $\E X = 0$ to simplify the first order term.
Now, fix $\delta > 0$ to be specified later. Thanks to \eqref{pf:taylor}, to prove \eqref{pf:integralcriterion3} it suffices to show that
\begin{align}
\limsup_{\delta\downarrow 0}\limsup_{n\rightarrow\infty}\sup_{\lambda \in [0,\delta]}\sup_{X \in \Phi}\E\left[A_n(\lambda,X) - B_n(\lambda,X)\right] < 0, \label{pf:integralcriterion11}
\end{align}
for sufficiently large $n$. 
Fix $\epsilon,\delta > 0$, with $\delta$ to be determined later.
If $\lambda \le \delta$, then since $\ell''$ is nondecreasing we have
\begin{align*}
\E[A_n(\lambda,X)] &\le 2\E\left[\left(\frac{X^2}{n^2} + |\gamma'(t_\lambda)|^2\right)\ell''(t_\lambda X/n - \gamma(t_\lambda))\right] \\
	&\le 2\E\left[\left(\frac{X^2}{n^2} + |\gamma'(\delta)|^2\right)\ell''(\delta X/n)\right].
\end{align*}
Now note that by assumption (3c) there exists $x_0 > 0$ such that $x^2\ell''(x) \le \epsilon\ell(\gamma^*(x))$ for all $x \ge x_0$.
Thus
\begin{align*}
\frac{X^2}{n^2}\ell''(\delta X/n)1_{|X| \ge x_0} &\le \frac{\epsilon}{\delta^2}\ell(\gamma^*(\delta|X|/n))1_{|X| \ge x_0} \le \frac{\epsilon}{\delta \kappa n}\ell(\gamma^*(\kappa|X|)) + \frac{\epsilon}{\delta^2}\ell(\gamma^*(\delta|X|/n))\left(1 - \frac{\delta}{\kappa n}\right),
\end{align*}
where the last inequality follows from convexity of $\ell \circ \gamma^*$ as in \eqref{pf:integralcriterion5} and holds whenever $n\kappa \ge \delta$.
On the other hand, 
\[
\frac{X^2}{n^2}\ell''(\delta X/n)1_{|X| < x_0} \le  \frac{x_0^2}{n^2}\ell''(\delta x_0/n),
\]
which clearly tends to zero as $n\rightarrow\infty$.
Combining the above two inequalities yields
\[
\limsup_{n \rightarrow \infty}\sup_{X \in \Phi}\E\left[\frac{X^2}{n^2}\ell''(\delta X/n)\right] \le \epsilon/\delta.
\]
Since $\epsilon > 0$ was arbitrary, the above limsup is in fact a limit equal to zero.
Thus
\[
\limsup_{n \rightarrow \infty}\sup_{X \in \Phi}\E[A_n(\lambda,X)] \le 2|\gamma'(\delta)|\sup_{X \in \Phi}\E[\ell''(\delta|X|)],
\]
A similar truncation argument (using $x_0$) shows that the supremum over $X \in \Phi$ on the right-hand side is finite, and thus since $\gamma'(\delta)=0$ we get
\begin{align}
\limsup_{\delta\downarrow 0}\limsup_{n \rightarrow \infty}\sup_{X \in \Phi}\E[A_n(\lambda,X)] = 0. \label{pf:integralcriterion12}
\end{align}
On the other hand, for the $B$ term, note that
\begin{align*}
\E[B_n(\lambda,X)] &\ge \gamma''(t_\delta)\E\left[\ell'\left(-\frac{\delta |X|}{n} - \gamma(\delta)\right)\right] \ge I_\delta  \E\left[\ell'\left(-\frac{\delta |X|}{n} - \gamma(\delta)\right)\right],
\end{align*}
where $I_\delta := \inf_{t \in [0,\delta]}\gamma''(t)$ is strictly positive for sufficiently small $\delta$. Thanks to \eqref{pf:integralcriterion7} from Case 2 above, we have
\[
\liminf_{n\rightarrow\infty}\sup_{X \in \Phi}\E[B_n(\lambda,X)] \ge I_\delta\ell'(-\gamma(\delta)),
\]
and thus
\[
\liminf_{\delta \downarrow 0}\liminf_{n\rightarrow\infty}\sup_{X \in \Phi}\E[B_n(\lambda,X)] \ge \gamma'(0)\ell'(0).
\]
Finally, from this and \eqref{pf:integralcriterion12} we conclude 
\begin{align*}
\limsup_{\delta\downarrow 0}\limsup_{n\rightarrow\infty}\sup_{\lambda \in [0,\delta]}\sup_{X \in \Phi}\E\left[A_n(\lambda,X) - B_n(\lambda,X)\right] \le -\gamma''(0)\ell'(0) < 0.
\end{align*}

\textbf{Case 3, Step 2:} To deal with large values of $\lambda$, we follow exactly the same proof as in Step 2 of Case 2 to show that there exists $N > 0$ such that \eqref{pf:integralcriterion8} holds.
Combining this with \eqref{pf:integralcriterion3} completes the proof of \eqref{pf:integralcriterion1}, and thus the theorem.

\bibliographystyle{amsplain}
\bibliography{../riskmeasures-bib}

\providecommand{\bysame}{\leavevmode\hbox to3em{\hrulefill}\thinspace}
\providecommand{\MR}{\relax\ifhmode\unskip\space\fi MR }
\providecommand{\MRhref}[2]{%
  \href{http://www.ams.org/mathscinet-getitem?mr=#1}{#2}
}
\providecommand{\href}[2]{#2}
\begin{thebibliography}{10}

\bibitem{acerbi-scandolo-liquidity}
C.~Acerbi and G.~Scandolo, \emph{Liquidity risk theory and coherent measures of
  risk}, Quantitative Finance \textbf{8} (2008), no.~7, 681--692.

\bibitem{artzner1999coherent}
P.~Artzner, F.~Delbaen, J.-M. Eber, and D.~Heath, \emph{Coherent measures of
  risk}, Mathematical Finance \textbf{9} (1999), no.~3, 203--228.

\bibitem{barrieu-elkaroui-infconvolution}
P.~Barrieu and N.~El Karoui, \emph{Inf-convolution of risk measures and optimal
  risk transfer}, Finance and Stochastics \textbf{9} (2005), no.~2, 269--298.

\bibitem{beiglbock-henrylabordere-penkner}
M.~Beiglb{\"o}ck, P.~Henry-Labord{\`e}re, and F.~Penkner,
  \emph{Model-independent bounds for option prices-a mass transport approach},
  Finance and Stochastics \textbf{17} (2013), no.~3, 477--501.

\bibitem{beiglbock-nutz-touzi}
N.~Beiglb{\"o}ck, M.~Nutz, and N.~Touzi, \emph{Complete duality for martingale
  optimal transport on the line}, arXiv preprint arXiv:1507.00671 (2015).

\bibitem{beiglboeck-juillet}
M.~Beiglboeck and N.~Juillet, \emph{On a problem of optimal transport under
  marginal martingale constraints}, arXiv preprint arXiv:1208.1509 (2012).

\bibitem{bental-teboulle-1986}
A.~Ben-Tal and M.~Teboulle, \emph{Expected utility, penalty functions, and
  duality in stochastic nonlinear programming}, Management Science \textbf{32}
  (1986), no.~11, 1445--1466.

\bibitem{bental-teboulle-2007}
\bysame, \emph{An old-new concept of convex risk measures: {T}he optimized
  certainty equivalent}, Mathematical Finance \textbf{17} (2007), no.~3,
  449--476.

\bibitem{biagini-frittelli}
S.~Biagini and M.~Frittelli, \emph{On the extension of the {N}amioka-{K}lee
  theorem and on the {F}atou property for risk measures}, Optimality and
  risk-modern trends in mathematical finance, Springer, 2010, pp.~1--28.

\bibitem{biane-voiculescu}
P.~Biane and D.~Voiculescu, \emph{A free probability analogue of the
  {W}asserstein metric on the trace-state space}, Geometric \& Functional
  Analysis GAFA \textbf{11} (2001), no.~6, 1125--1138.

\bibitem{bobkov-ding}
S.~Bobkov and Y.~Ding, \emph{Optimal transport and {R}{\'e}nyi informational
  divergence}, Preprint (2014).

\bibitem{bobkov-gotze}
S.~Bobkov and F.~G{\"o}tze, \emph{Exponential integrability and transportation
  cost related to logarithmic {S}obolev inequalities}, Journal of Functional
  Analysis \textbf{163} (1999), no.~1, 1--28.

\bibitem{cheridito-li-orlicz}
P.~Cheridito and T.~Li, \emph{Risk measures on {O}rlicz hearts}, Mathematical
  Finance \textbf{19} (2009), no.~2, 189--214.

\bibitem{cordero2002some}
D.~Cordero-Erausquin, \emph{Some applications of mass transport to
  {G}aussian-type inequalities}, Archive for rational mechanics and analysis
  \textbf{161} (2002), no.~3, 257--269.

\bibitem{ding2014wasserstein}
Y.~Ding, \emph{Wasserstein-{D}ivergence transportation inequalities and
  polynomial concentration inequalities}, Statistics \& Probability Letters
  \textbf{94} (2014), 77--85.

\bibitem{feyel-ustunel}
D.~Feyel and A.S. {\"U}st{\"u}nel, \emph{Monge-{K}antorovitch measure
  transportation and {M}onge-{A}mpere equation on {W}iener space}, Probability
  theory and related fields \textbf{128} (2004), no.~3, 347--385.

\bibitem{filipovic-svindland}
D.~Filipovi{\'c} and G.~Svindland, \emph{The canonical model space for
  law-invariant convex risk measures is ${L}^1$}, Mathematical Finance
  \textbf{22} (2012), no.~3, 585--589.

\bibitem{follmer-schied-convex}
H.~F{\"o}llmer and A.~Schied, \emph{Convex measures of risk and trading
  constraints}, Finance and stochastics \textbf{6} (2002), no.~4, 429--447.

\bibitem{follmer-schied-book}
\bysame, \emph{Stochastic finance: {A}n introduction in discrete time}, Walter
  de Gruyter, 2011.

\bibitem{frittelli2002putting}
M.~Frittelli and E.R. Gianin, \emph{Putting order in risk measures}, Journal of
  Banking \& Finance \textbf{26} (2002), no.~7, 1473--1486.

\bibitem{gozlan-leonard-largedeviation}
N.~Gozlan and C.~L{\'e}onard, \emph{A large deviation approach to some
  transportation cost inequalities}, Probability Theory and Related Fields
  \textbf{139} (2007), no.~1-2, 235--283.

\bibitem{gozlan-leonard-survey}
\bysame, \emph{Transport inequalities. {A} survey}, arXiv preprint
  arXiv:1003.3852 (2010).

\bibitem{guillin-leonard-wu-yao}
A.~Guillin, C.~L{\'e}onard, L.~Wu, and N.~Yao, \emph{Transportation-information
  inequalities for {M}arkov processes}, Probability theory and related fields
  \textbf{144} (2009), no.~3-4, 669--695.

\bibitem{hamel-heyde-rudloff}
A.~Hamel, F.~Heyde, and B.~Rudloff, \emph{Set-valued risk measures for conical
  market models}, Mathematics and financial economics \textbf{5} (2011), no.~1,
  1--28.

\bibitem{jarrow-protter-liquidity}
R.~Jarrow and P.~Protter, \emph{Liquidity risk and risk measure computation},
  Review of Futures Markets \textbf{11} (2005), no.~1, 27--39.

\bibitem{jouini-meddeb-touzi}
E.~Jouini, M.~Meddeb, and N.~Touzi, \emph{Vector-valued coherent risk
  measures}, Finance and Stochastics \textbf{8} (2004), no.~4, 531--552.

\bibitem{jouini-touzi-schachermayer}
E.~Jouini, W.~Schachermayer, and N.~Touzi, \emph{Law invariant risk measures
  have the {F}atou property}, Advances in mathematical economics, Springer,
  2006, pp.~49--71.

\bibitem{kaina-ruschendorf}
M.~Kaina and L.~R{\"u}schendorf, \emph{On convex risk measures on
  ${L}^p$-spaces}, Mathematical methods of operations research \textbf{69}
  (2009), no.~3, 475--495.

\bibitem{kupper-schachermayer}
M.~Kupper and W.~Schachermayer, \emph{Representation results for law invariant
  time consistent functions}, Mathematics and Financial Economics \textbf{2}
  (2009), no.~3, 189--210.

\bibitem{lacker-lawinvariant}
D.~Lacker, \emph{Law invariant risk measures and information divergences},
  Preprint (2015).

\bibitem{ledoux-concentration}
M.~Ledoux, \emph{The concentration of measure phenomenon}, no.~89, American
  Mathematical Soc., 2005.

\bibitem{marton-blowingup}
K.~Marton, \emph{A simple proof of the blowing-up lemma}, Information Theory,
  IEEE Transactions on \textbf{32} (1986), no.~3, 445--446.

\bibitem{marton-bounding}
\bysame, \emph{Bounding $\bar{d}$-distance by informational divergence: a
  method to prove measure concentration}, The Annals of Probability \textbf{24}
  (1996), no.~2, 857--866.

\bibitem{rockafellar2006generalized}
R.T. Rockafellar, S.~Uryasev, and M.~Zabarankin, \emph{Generalized deviations
  in risk analysis}, Finance and Stochastics \textbf{10} (2006), no.~1, 51--74.

\bibitem{svindland-continuity}
G.~Svindland, \emph{Continuity properties of law-invariant (quasi-) convex risk
  functions on ${L}^\infty$}, Mathematics and Financial Economics \textbf{3}
  (2010), no.~1, 39--43.

\bibitem{svindland-convexorder}
\bysame, \emph{A note on law invariance and convex order monotonicity},
  (2013).

\bibitem{talagrand-transportation}
M.~Talagrand, \emph{Transportation cost for {G}aussian and other product
  measures}, Geometric \& Functional Analysis GAFA \textbf{6} (1996), no.~3,
  587--600.

\bibitem{villani-book}
C.~Villani, \emph{Topics in optimal transportation}, no.~58, American
  Mathematical Soc., 2003.

\bibitem{weber-liquidityadjusted}
S.~Weber, W.~Anderson, A.-M. Hamm, T.~Knispel, M.~Liese, and T.~Salfeld,
  \emph{Liquidity-adjusted risk measures}, Mathematics and Financial Economics
  \textbf{7} (2013), no.~1, 69--91.

\end{thebibliography}

\end{document}